\documentclass{fundam}

\publyear{22}
\papernumber{2102}
\volume{185}
\issue{1}

\usepackage{graphicx}
\usepackage[algo2e,ruled,lined]{algorithm2e}
\usepackage{color}
\usepackage{amssymb,amsmath,stmaryrd} 
\usepackage{xcolor,pgf,tikz}
\usepackage{circuitikz}
\usepackage{mathtools}

\usetikzlibrary{automata,patterns,topaths,shapes,calc,decorations,decorations.pathmorphing,decorations.pathreplacing,fit,arrows,petri,fadings,positioning}

\usepackage{subfig}

\usetikzlibrary{arrows,automata,shapes,snakes,calc,positioning,petri}
\usepackage{tcolorbox}
\usepackage{soul}
\usepackage[shortlabels]{enumitem}

\usepackage{array}
\usepackage{booktabs}

\usepackage{thmtools}
\usepackage{thm-restate}

\usepackage{comment}

\usepackage{enumitem}
\setlist[itemize]{noitemsep, topsep=0pt}

\def\QEDclosed{\mbox{\rule[0pt]{1.3ex}{1.3ex}}} 

\def\QED{\QEDclosed} 

\newcommand{\nat}{\mathbb{N}}
\newcommand{\integer}{\mathbb{Z}}
\newcommand{\rat}{\mathbb{Q}}
\newcommand{\real}{\mathbb{R}}

\newcommand{\Front}{{\sf Front}} 

\newcommand{\Val}{{\texttt{VH}}}

\newcommand{\prob}{{\bf Pr}}
\newcommand{\esp}{{\bf E}}
\newcommand{\suiv}{p}

\newenvironment{itemize*}%
  {\begin{itemize}%
    \setlength{\itemsep}{0pt}%
    \setlength{\parskip}{0pt}}%
  {\end{itemize}}

\begin{document}
\title{Introducing Divergence for Infinite Probabilistic Models\thanks{This work of Alain Finkel has been supported by ANR project BRAVAS (ANR-17-CE40-0028). This work of Serge Haddad and Lina Ye has been supported by ANR project MAVeriQ (ANR-20-CE25-0012).}}

\author{Alain Finkel\\
Université Paris-Saclay, CNRS, ENS Paris-Saclay, IUF, LMF \\
91190 Gif-sur-Yvette, France\\
alain.finkel{@}lmf.cnrs.fr
\and Serge Haddad\\
Université Paris-Saclay, CNRS, ENS Paris-Saclay, LMF \\
91190 Gif-sur-Yvette, France\\
serge.haddad{@}lmf.cnrs.fr
\and Lina Ye\\
Université Paris-Saclay, CNRS, ENS Paris-Saclay, CentraleSupélec, LMF \\
91190 Gif-sur-Yvette, France\\
lina.ye{@}lmf.cnrs.fr
}

\maketitle              
\begin{abstract}
Computing the reachability probability in infinite state probabilistic models has been the topic of numerous works.
Here we introduce a new property called \emph{divergence} that when satisfied  allows to compute reachability probabilities up to an arbitrary precision.
One of the main interests of divergence is that this computation does not require the reachability problem, i.e., the possibility to reach target states 
from an initial state in a given model,  to be decidable.
Then we study the decidability of divergence for random walks and the probabilistic versions of Petri nets
where the weights associated with transitions may also depend on the current state.
This should be contrasted with most of the existing works that assume 
weights independent of the state.
Such an extended framework is motivated by the modeling of real case studies.
Moreover, we exhibit some subclasses of  channel systems and pushdown automata that are divergent by construction, 
particularly suited for specifying open distributed systems and networks prone to performance collapsing where probabilities related to service
requirements are needed.  
\end{abstract}
\section{Introduction}
\label{sec:introduction}

\noindent
{\bf Probabilistic models.}
In the 1980's, finite-state Markov chains have been considered for the modeling and analysis 
of probabilistic concurrent finite-state programs~\cite{DBLP:conf/focs/Vardi85}.
\setcounter{footnote}{0} 
Since the 2000's, many works have been done to study the infinite-state Markov chains 
obtained from probabilistic versions 
of automata extended with unbounded data (like stacks, channels, counters and clocks)\footnote{In 1972,  Santos gave the first definition of \emph{probabilistic pushdown automata} \cite{DBLP:journals/iandc/Santos72c} (to the best of our knowledge) 
that surprisingly did not open up a new field of research at that time.}.
The (qualitative and quantitative) model checking of \emph{probabilistic pushdown automata} (pPDA) is studied in many papers~\cite{DBLP:conf/lics/EsparzaKM04,DBLP:conf/lics/EsparzaKM05,DBLP:conf/focs/BrazdilEK05,Esparza06} (see \cite{DBLP:journals/fmsd/BrazdilEKK13} for a survey). In 1997, Iyer and Narasimha \cite{DBLP:conf/tapsoft/IyerN97} started the analysis of \emph{probabilistic lossy channel systems} (pLCS) and later both some qualitative and quantitative properties were shown 
decidable for pLCS \cite{DBLP:journals/iandc/AbdullaBRS05}.
\emph{Probabilistic counter machines} (pCM) have also been investigated~\cite{DBLP:conf/csl/BrazdilKKNK14,DBLP:journals/jacm/BrazdilKK14,DBLP:conf/cav/BrazdilKK11}.

\smallskip\noindent
{\bf Computing the reachability probability.}
In finite Markov chains, there is a well-known algorithm for computing  the reachability probabilities in
 polynomial time~\cite{ModelChecking08}. Here we focus on the problem of \emph{Computing the Reachability Probability up to an arbitrary precision} (CRP) in \emph{infinite} Markov chains. 
 To do so, there are (at least) two possible research directions. 

The first one is to consider the Markov chains associated with a particular class of probabilistic models (like pPDA or probabilistic Petri nets (pPN))
and some specific target sets 
and to exploit the properties of these models to design a CRP-algorithm.
For instance in~\cite{DBLP:journals/fmsd/BrazdilEKK13}, the authors 
exhibit a PSPACE algorithm for pPDA and PTIME algorithms for single-state pPDA and for one-counter automata.

The second one consists in exhibiting a property of Markov chains that yields a generic algorithm
for solving the CRP problem and then looking for models that generate Markov chains that fulfill this property.
\emph{Decisiveness} of Markov chains is such a property. In words, decisiveness w.r.t. $s_0$ and $A$ means that almost surely the random path $\sigma$ starting from $s_0$ will reach $A$ 
or some state $s'$ from which $A$ is unreachable. It has been shown that pPDA are not (in general) decisive but both pLCS and probabilistic Petri nets (pPN) are decisive (for pPN: when the target set is upward-closed~\cite{AbdullaHM07}). 

\smallskip\noindent
{\bf Two limits of the previous approaches.}
The generic approach based on decisiveness property has numerous applications but suffers from the restriction that the reachability problem must be decidable
in the corresponding non deterministic model.
To the best of our knowledge, all generic approaches rely on a \emph{decidable reachability problem}. 

In most of the works, the probabilistic models associate with  transitions a \emph{constant} weight 
and get transition probabilities by normalizing these weights among the enabled transitions in the current state. 
This \emph{forbids to model phenomena} like congestion in networks (resp. performance collapsing in distributed systems) 
when the number of messages (resp. processes) exceeds some threshold leading to an increasing probability of message arrivals (resp. process creations)
before message departures (resp. process terminations).

\smallskip\noindent
{\bf Our contributions.}
\begin{itemize}
	\item  In order to handle realistic phenomena (like congestion in networks), one considers
	\emph{dynamic} weights i.e., weights depending on the current state.
	\item We introduce the new \emph{divergence} property of  Markov chains w.r.t. $s_0$ and $A$:  given some precision $\theta$, one can discard a set of states
	with either a small probability to be reached from $s_0$ or a small probability to reach $A$ such that the remaining subset of states is finite
	and thus allows for an approximate computation of the reachability probability up to $\theta$. For divergent Markov chains, we provide a generic algorithm for the CRP-problem
	that \emph{does not require} the decidability of the reachability problem. Furthermore, another generic algorithm is also proposed when the reachability problem is decidable. 
	While decisiveness and divergence are not exclusive (both hold for finite Markov chains), 
	they are in some sense complementary. 
	Indeed, divergence is somehow related to transience of Markov chains while decisiveness is somehow related to recurrence \cite{Alain23}. 
	\item We study the decidability of divergence for different models. Our first undecidability result implies that whatever the infinite models, 
	one must restrict the kind of dynamics weights. Here we limit them to polynomial weights.
	\item We prove, by case analysis, that divergence is decidable for a subclass of random walks (i.e. random walks with \emph{polynomial} weights). 
	Furthermore, we show that divergence is undecidable for polynomial pPNs  w.r.t. an upward closed set.
	\item	In order to check divergence, we provide several simpler sufficient conditions based on 
	martingale theory. 
	\item We provide two classes of divergent polynomial models. 
	The first one is a probabilistic version of channel systems 
	particularly suited for the modeling of open queuing networks. 
	The second one is the probabilistic version of pushdown automata restricted to some typical behaviors of dynamic systems.
\end{itemize}

\smallskip\noindent
{\bf Organisation.} Section~\ref{sec:divergence} recalls Markov chains, introduces  divergent Markov chains, 
and presents two algorithms for solving the CRP-problem.  
In Section~\ref{sec:rw} and Section~\ref{sec:pPN}, we study the decidability status of divergence for random walks and pPNs. 
Finally Section~\ref{sec:case} presents two divergent subclasses of  
probabilistic channel systems  and pPDA. Finally we conclude and give  perspectives to this work in Section~\ref{sec:conclusion}.
Some technical proofs can be found in Appendix.

\section{Divergence of  Markov chains}
\label{sec:divergence}

\subsection{Markov chains: definitions and properties}

\noindent
{\bf Notations.} 
A set $S$ is \emph{countable} if there exists an injective function from $S$ 
to  the set of natural numbers: hence it could be finite or countably infinite.
Let $S$ be a countable set of elements called states.
Then $Dist(S) = \{\Delta : S \rightarrow \real_{\geq 0} \mid \sum_{s\in S}\Delta(s)=1\}$  is the set of \emph{distributions} over $S$.
Let $\Delta \in Dist(S)$, then $Supp(\Delta)=\Delta^{-1}(\real_{>0})$.
Let $T\subseteq S$, then  $S\setminus T$ will also be denoted $\overline{T}$.

\begin{definition}[Effective Markov chain]
A \emph{Markov chain} $\mathcal M=(S,\suiv)$ is a tuple where:
\begin{itemize}  
	\item  $S$ is a countable set of states,
	\item $\suiv$ is the transition function from $S$ to $Dist(S)$;
\end{itemize}
When for all $s\in S$, (1) $Supp(\suiv(s))$ is finite and computable and (2) the function $s \mapsto \suiv(s)$ is computable,
one says that  $\mathcal M$ is \emph{effective}.
\end{definition}

\noindent
{\bf Notations.} $\suiv^{(d)}$ is the $d^{th}$ power of the transition matrix $\suiv$.
When $S$ is countably infinite, we say that $\mathcal M$ is \emph{infinite} and we sometimes identify $S$ with $\nat$.
We also denote $\suiv(s)(s')$ by $\suiv(s,s')$
and $\suiv(s,s')>0$ by
$s\xrightarrow{\suiv(s,s')} s'$.  A Markov chain is also viewed as a transition system
whose transition relation $\rightarrow$ is defined by $s\rightarrow s'$ if $\suiv(s,s')>0$. 
Let $A \subseteq S$, one denotes $Post_{\mathcal M}^*(A)$, the set of states
that can  be reached from some state of $A$ and $Pre_{\mathcal M}^*(A)$, the set of states
that can reach $A$. As usual, we denote $\rightarrow^*$,
the transitive closure of $\rightarrow$ and we say that
$s'$ is  \emph{reachable from $s$} if  $s \rightarrow^* s'$. 
We say that a subset $A \subseteq S$ is \emph{reachable} from $s$ if some $s'\in A$ is reachable from $s$.
Note that every finite path of 
$\mathcal M$ can be extended into (at least) one infinite path.

\begin{example}
Let $\mathcal M_1$ be the Markov chain of Figure~\ref{fig:rw}.
In any state $i>0$, the probability for going to the ``right'', $p(i,i+1)$, is equal to $0<p_i<1$ and  for going to the ``left''  $p(i,i-1)$
is equal to $1-p_i$.
In state $0$, one goes to  $1$ with probability 1. $\mathcal M_1$ is effective if the function $n\mapsto p_n$
is computable.
\end{example}

\begin{figure}
\begin{center}
  \begin{tikzpicture}[node distance=2cm,->,auto,-latex,scale=0.9]
    
   \path (0,0) node[minimum size=0.6cm,draw,circle,inner sep=2pt] (q0) {$0$};
   \path (2,0) node[minimum size=0.6cm,draw,circle,inner sep=2pt] (q1) {$1$};
   \path (4,0) node[minimum size=0.6cm,draw,circle,inner sep=2pt] (q2) {$2$};
   \path (6,0) node[minimum size=0.6cm,draw,circle,inner sep=2pt] (q3) {$3$};
   \path (8,0) node[] (q4) {$\cdots$};

    \draw[arrows=-latex'] (q0)-- (1,0.5) node[pos=1,above] {$1$}--(q1) ;
    \draw[arrows=-latex'] (q1)-- (1,-0.5) node[pos=1,below] {$1-p_1$}--(q0) ;

    \draw[arrows=-latex'] (q1)-- (3,0.5) node[pos=1,above] {$p_1$}--(q2) ;
    \draw[arrows=-latex'] (q2)-- (3,-0.5) node[pos=1,below] {$1-p_2$}--(q1) ;

    \draw[arrows=-latex'] (q2)-- (5,0.5) node[pos=1,above] {$p_2$}--(q3) ;
    \draw[arrows=-latex'] (q3)-- (5,-0.5) node[pos=1,below] {$1-p_3$}--(q2) ;

    \draw[arrows=-latex'] (q3)-- (7,0.5) node[pos=1,above] {$p_3$}--(q4) ;
    \draw[arrows=-latex'] (q4)-- (7,-0.5) node[pos=1,below] {$1-p_4$}--(q3) ;

  \end{tikzpicture}
\end{center}
\caption{A random walk $\mathcal M_1$}
\label{fig:rw}
\end{figure}

\smallskip
Given an initial state $s_0$, the \emph{sampling} of a Markov chain $\mathcal M$ is an \emph{infinite
random sequence of states} (i.e., a path) $\sigma=s_0s_1\ldots$ such that for all $i\geq 0$,
$s_i\rightarrow s_{i+1}$. As usual,  the corresponding $\sigma$-algebra whose items are called events
is generated by the finite prefixes of infinite paths and the probability of an event $Ev$ given an initial state $s_0$ is denoted $\prob_{\mathcal M,s_0}(Ev)$. In case of a finite path $s_0\ldots s_n$,
$\prob_{\mathcal M,s_0}(s_0\ldots s_n)=\prod_{0 \leq i<n} \suiv(s_i,s_{i+1})$.

\smallskip
\noindent {\bf Notations.}
From now on,  {\bf G} (resp. {\bf F}, {\bf X}) denotes the always (resp. eventual, next)
operator of LTL.   

\smallskip
Let $A\subseteq S$ and $\sigma$ be an infinite path. We say that $\sigma$ 
\emph{reaches} $A$ if $\exists i\in \nat\ s_i\in A$ and that  $\sigma$ 
\emph{visits} $A$ if $\exists i>0\ s_i\in A$. 
The probability that starting from $s_0$, the path $\sigma$ reaches (resp. visits) $A$ will be denoted
by $\prob_{\mathcal M,s_0}({\bf F}  A)$ (resp. $\prob_{\mathcal M,s_0}({\bf XF}  A)$).

\smallskip
We now state qualitative and quantitative properties of a Markov chain.
\begin{definition}[Irreducibility, recurrence, transience] Let $\mathcal M=(S,\suiv)$ be a Markov chain and  $s \in S$. Then
 $\mathcal M$ is \emph{irreducible} if for all $s,s'\in S$, $s\rightarrow^* s'$.
$s$ is \emph{recurrent} if $\prob_{\mathcal M,s}({\bf XF}  \{s\})=1$ otherwise $s$ is \emph{transient}.

\end{definition}

In an irreducible Markov chain, all states are in the same category, either recurrent or transient~\cite{KSK76}. 
Thus an irreducible Markov chain will be said transient or 
recurrent depending on the category of its states. 
In the remainder of this section,
we will relate this category with techniques for computing reachability
probabilities.

\begin{example}
$\mathcal M_1$ is clearly irreducible. $\mathcal M_1$ is recurrent if and only if  
$\sum_{n\in \nat} \prod_{1\leq m\leq n} \rho_m =\infty$ with $\rho_m=\frac{1-p_m}{p_m}$, and when transient, the probability that starting from $i$ the random path visits
$0$ is equal to  $\frac{\sum_{i\leq n} \prod_{1\leq m\leq n} \rho_m}{\sum_{n\in \nat} \prod_{1\leq m\leq n} \rho_m}$ (see~\cite{norris97} for more details). 

\end{example}

One of our goals is to approximately compute reachability probabilities in infinite Markov 
chains. 
Let us formalize it. Given a finite representation of a subset $A\subseteq S$,
one says that this representation is \emph{effective} if one can decide the membership problem for $A$. 
With a slight abuse of language, we
identify $A$ with any effective representation of $A$.

\smallskip
\centerline{ {\bf \small{The Computing of Reachability Probability  (CRP) problem}} }
\begin{center}
\fbox{
\begin{minipage}{0.90\textwidth}

\noindent
$\bullet$ Input: an effective Markov chain $\mathcal M$, an (initial) state $s_0$, 
an effective subset of states $A$, and a rational $\theta>0$.

\noindent
$\bullet$ 
 Output: an interval $[low,up]$ such that $up-low\leq \theta$ and $\prob_{\mathcal M,s_0}({\bf F}  A) \in [low,up]$.
\end{minipage}
}
\end{center}

\vspace{3mm}
\subsection{Divergent Markov chains}
 
Let us first discuss two examples before introducing the notion of \emph{divergent} Markov chains. 

\begin{example} 
\label{example:twochains}Consider again the Markov chain $\mathcal M_1$ of Figure~\ref{fig:rw} with for all $n>0$,
$p_n=p > \frac 1 2$. 
In this case, for $m\geq 0$,
$\prob_{\mathcal M_1,m}({\bf F} \{0\})=\rho^m$ with $\rho=\frac {1-p}{p}$. Thus here the key point is
that not only this reachability probability is less than 1 but it goes to 0 when $m$ goes to $\infty$. 
This means that given some precision $\theta$, one could ``prune''  states $n\geq n_0$
and compute the reachability probabilities of $A$ in a finite Markov chain.

\noindent
Consider the Markov chain  $\mathcal M_2$ of Figure~\ref{fig:divrw}. If $\prob_{\mathcal M,0}({\bf F} \{m,m+1,\ldots\})=\prod_{n< m}p_n$
goes to $0$ when $m$ goes to $\infty$, then given some precision $\theta$, one could ``prune''  states $n\geq n_0$
and compute the reachability probabilities of $A$ in a finite Markov chain.
\end{example}

\begin{figure}
\begin{center}
  \begin{tikzpicture}[node distance=2cm,->,auto,-latex,scale=0.9]
  
  \draw (-0.5,-1.5) -- (8.5,-1.5) -- (8.5,-2.3) -- (-0.5,-2.3) -- cycle;
    
   \path (0,-0.5) node[minimum size=0.6cm,draw,circle,inner sep=2pt] (q0) {$0$};
   \path (2,-0.5) node[minimum size=0.6cm,draw,circle,inner sep=2pt] (q1) {$1$};
   \path (4,-0.5) node[minimum size=0.6cm,draw,circle,inner sep=2pt] (q2) {$2$};
   \path (6,-0.5) node[minimum size=0.6cm,draw,circle,inner sep=2pt] (q3) {$3$};
   \path (8,-0.5) node[] (q4) {$\cdots$};

  \path (4,-2) node[] (q5) {A finite Markov chain containing $A$};


    \draw[arrows=-latex'] (q0) --(q1)node[pos=0.5,above] {$p_0$} ;
    \draw[arrows=-latex'] (q0) --(0,-1.8)node[pos=0.3,right] {$1-p_0$} ;
    
    \draw[arrows=-latex'] (q1) --(2,-1.8) node[pos=0.3,right] {$1-p_1$};
    \draw[arrows=-latex'] (q1) --(q2) node[pos=0.5,above] {$p_1$};
    
    \draw[arrows=-latex'] (q2)-- (4,-1.8)  node[pos=0.3,right] {$1-p_2$};
    \draw[arrows=-latex'] (q2) --(q3) node[pos=0.5,above] {$p_2$} ;

    \draw[arrows=-latex'] (q3)-- (6,-1.8) node[pos=0.3,right] {$1-p_3$};
    \draw[arrows=-latex'] (q3)-- (q4) node[pos=0.5,above] {$p_3$};

  \end{tikzpicture}
\end{center}
\caption{An infinite Markov chain $\mathcal M_2$}
\label{fig:divrw}
\end{figure}

In words, a divergent Markov chain w.r.t. $s_0$ and $A$ generalizes these examples: given some precision $\theta$, one can discard a set of states
with either a small probability to be reached from $s_0$ or a small probability to reach $A$ such that the remaining subset of states is finite
and thus allows for an approximate computation of the reachability probability up to $\theta$.

\begin{definition}[divergent Markov chain]
\label{definition:divergent}
Let $\mathcal M$ be a Markov chain, $s_0\in S$ 
and $A\subseteq S$. 
We say that $\mathcal M$ 
is \emph{divergent} w.r.t. $s_0$ and $A$ if there exist two computable functions $f_0$ and $f_1$
from $S$ to $\real_{\geq 0}$ such that:
\begin{itemize}
	\item For all $0<\theta<1$, $\prob_{\mathcal M,s_0}(\mathbf{F} f_0^{-1}([0,\theta]))\leq \theta$;
	\item For all $s\in S$, $\prob_{\mathcal M,s}(\mathbf{F} A)\leq f_1(s)$;
	\item For all $0<\theta<1$, $\{s \mid f_0(s) \geq \theta \wedge  f_1(s) \geq \theta\} \cap Post_{\mathcal M}^*(\{s_0\})$ is finite.
\end{itemize}
\end{definition}

\begin{figure}[h]
\begin{center}
\tikzset{every picture/.style={line width=0.75pt}} 

\begin{tikzpicture}[x=0.75pt,y=0.75pt,yscale=-0.9,xscale=0.9]

\draw  [fill={rgb, 255:red, 246; green, 235; blue, 235 }  ,fill opacity=1 ] (202.44,55) -- (429.56,55) .. controls (459.07,55) and (483,96.19) .. (483,147) .. controls (483,197.81) and (459.07,239) .. (429.56,239) -- (202.44,239) .. controls (172.93,239) and (149,197.81) .. (149,147) .. controls (149,96.19) and (172.93,55) .. (202.44,55) -- cycle ;
\draw [line width=3.75]  [dash pattern={on 4.22pt off 3.52pt}]  (197,99) -- (399,99) ;
\draw [shift={(409,99)}, rotate = 180] [fill={rgb, 255:red, 0; green, 0; blue, 0 }  ][line width=0.08]  [draw opacity=0] (27.6,-6.9) -- (0,0) -- (27.6,6.9) -- cycle    ;
\draw  [dash pattern={on 0.84pt off 2.51pt}]  (405,165) -- (298.86,199.08) ;
\draw [shift={(296,200)}, rotate = 342.2] [fill={rgb, 255:red, 0; green, 0; blue, 0 }  ][line width=0.08]  [draw opacity=0] (8.93,-4.29) -- (0,0) -- (8.93,4.29) -- cycle    ;
\draw  [dash pattern={on 0.84pt off 2.51pt}]  (404,180) -- (310.85,211.05) ;
\draw [shift={(308,212)}, rotate = 341.57] [fill={rgb, 255:red, 0; green, 0; blue, 0 }  ][line width=0.08]  [draw opacity=0] (8.93,-4.29) -- (0,0) -- (8.93,4.29) -- cycle    ;
\draw  [dash pattern={on 0.84pt off 2.51pt}]  (411,193) -- (315.86,223.1) ;
\draw [shift={(313,224)}, rotate = 342.45] [fill={rgb, 255:red, 0; green, 0; blue, 0 }  ][line width=0.08]  [draw opacity=0] (8.93,-4.29) -- (0,0) -- (8.93,4.29) -- cycle    ;
\draw  (263,224) .. controls (263,210.19) and (274.19,199) .. (288,199) .. controls (301.81,199) and (313,210.19) .. (313,224) .. controls (313,237.81) and (301.81,249) .. (288,249) .. controls (274.19,249) and (263,237.81) .. (263,224) -- cycle ;
\draw  [fill={rgb, 255:red, 184; green, 233; blue, 134 }  ,fill opacity=0.5 ] (414,109) .. controls (414,84.7) and (433.7,65) .. (458,65) .. controls (482.3,65) and (502,84.7) .. (502,109) .. controls (502,133.3) and (482.3,153) .. (458,153) .. controls (433.7,153) and (414,133.3) .. (414,109) -- cycle ;
\draw  [fill={rgb, 255:red, 80; green, 227; blue, 194 }  ,fill opacity=0.5 ] (406,170) .. controls (406,145.7) and (425.7,126) .. (450,126) .. controls (474.3,126) and (494,145.7) .. (494,170) .. controls (494,194.3) and (474.3,214) .. (450,214) .. controls (425.7,214) and (406,194.3) .. (406,170) -- cycle ;
\draw   (173,99) .. controls (173,92.37) and (178.37,87) .. (185,87) .. controls (191.63,87) and (197,92.37) .. (197,99) .. controls (197,105.63) and (191.63,111) .. (185,111) .. controls (178.37,111) and (173,105.63) .. (173,99) -- cycle ;

\draw (176,93) node [anchor=north west][inner sep=0.75pt]   [font=\footnotesize] {$s_{0}$};
\draw (420,99.4) node [anchor=north west][inner sep=0.75pt]  [font=\tiny]  {$\{s\mid f_{0}( s) < \theta )\}$};
\draw (412,164.4) node [anchor=north west][inner sep=0.75pt]  [font=\tiny]  {$\{s\mid f_{1}( s) < \theta )\}$};
\draw (279,216.4) node [anchor=north west][inner sep=0.75pt]   [font=\footnotesize] {$A$};
\draw (150,140) node [anchor=north west][inner sep=0.75pt]    {\tiny{$|\{s \mid f_0(s)  \geq  \theta \wedge  f_1(s)  \geq \theta\} \cap Post_{\mathcal M}^*(s_0)|<\infty$}};
\draw (90,60) node [anchor=north west][inner sep=0.75pt]    [font=\footnotesize] {$Post_{\mathcal M}^*(s_0)$};
\draw (236,78.3) node [anchor=north west][inner sep=0.75pt]  [font=\tiny]  {$Pr_{\mathcal{M} ,s_{0}}\left(Ff_{0}^{-1}([0,\ \theta])\right) \leq \theta$};
\draw (344,218.4) node [anchor=north west][inner sep=0.75pt]  [font=\tiny]  {$Pr_{\mathcal{M} ,s}( FA) < \theta$};

\end{tikzpicture}
\caption{Illustration of divergence: the subset of $Post_{\mathcal M}^*(s_0)$ with $f_0(s) \ge \theta$ and $f_1(s) \ge \theta$ is finite.}
 \label{fig:diver}

\end{center}
\end{figure}

\noindent
{\bf Observation and illustration.} Let us remark that there cannot exist, for general Markov chains, an algorithm to decide the existence of such functions $f_0,f_1$ and when they exist, to compute them. 
But as shown later on there exist some simple sufficient conditions for divergence. Figure~\ref{fig:diver} illustrates the notion of divergence: for any $\theta$, the set of reachable states from $s_0$ after discarding the ones related to the conditions specified by $f_0$ and $f_1$ is finite. Intuitively, states with $f_0(s) \le \theta$ can be 
pruned because they are unlikely to be 
reached; states with $f_1(s) \le \theta$ can be  
discarded because they are unlikely to reach A. 
Divergence can be considered as the complement of decisiveness in some particular cases, such as simple random walks.

A finite Markov chain is divergent (letting $f_0=f_1=1$) w.r.t. any $s_0$ and any $A$.
In the first Markov chain of Example~\ref{example:twochains}, $f_0=1$ and $f_1(m)=\rho^m$
and in the second Markov chain, $f_1=1$, $f_0(m)=\prod_{0\leq n<m}p_n$
and $f_0(s)=1$ for all $s$ in the finite Markov chain containing $A$.
Generalizing these two examples, the next proposition introduces a sufficient condition
for divergence. The intuition is very simple: one can resolve CRP problem if the remaining state space is finite by discarding either a set of states
with a small probability to be reached from $s_0$ or a set of states with a small probability to reach $A$. 
Its proof is immediate by choosing ($f=f_0$ and $f_1=1$) or ($f=f_1$ and $f_0=1$).
\begin{proposition}
\label{proposition:f01}
Let $\mathcal M$ be a Markov chain, $s_0\in S$, 
$A\subseteq S$, 
and a computable function $f$
from $S$ to $\real_{\geq 0}$ such that:
\begin{itemize}
	\item For all $0<\theta<1$, $\prob_{\mathcal M,s_0}(\mathbf{F} f^{-1}([0,\theta]))\leq \theta$
	or for all $s\in S$, $\prob_{\mathcal M,s}(\mathbf{F} A)\leq f(s)$;
	\item For all $0<\theta<1$, $\{s \mid f(s) \geq \theta\} \cap Post_{\mathcal M}^*(\{s_0\})$ is finite.
\end{itemize}
Then $\mathcal M$  is divergent w.r.t. $s_0$ and $A$. 
\end{proposition}

\noindent
{\bf Decisiveness vs. divergence.} Interestingly  in the random walk $\mathcal M_1$,  
decisiveness and divergence are complementary.
Let $p_n=p > \frac 1 2$ for all $n>0$ and $s_0=1$ and $A=\{0\}$.
Then $\mathcal M_1$ is not decisive but divergent w.r.t. $s_0$ and $A$. On the other hand, if $p_n=p \le \frac 1 2$ for all $n>0$, 
then $\mathcal M_1$ becomes decisive but not divergent w.r.t. $s_0$ and $A$. However, this property may not hold. 
Consider the Markov chain $\mathcal M_2$ of Figure~\ref{fig:divrw}, if $\prod_{n\in \nat}p_n=0$, then it is both decisive and divergent. 
Otherwise $\mathcal M_2$ becomes non decisive but still divergent as the remaining states after pruning 
the set of states whose probability to reach $A$ is arbitrarily small  is finite. 

\subsection{Two algorithms for divergent Markov chains}

We now design two algorithms for accurately bounding the reachability probability for a
divergent (effective) Markov chain w.r.t. $s_0$ and an effective $A$, based on whether 
the reachability problem is decidable or not. We do not discuss their complexity, which depends both
on the functions $f_0$ and $f_1$ and on the models to be studied.

 \begin{algorithm2e}
  \LinesNumbered
  \DontPrintSemicolon
  \SetKwFunction{CompProb}{CompProb}
  \SetKwFunction{CompFinProb}{CompFinProb}
 \SetKwFunction{PosReach}{PosReach}
  \SetKwFunction{Insert}{Insert}
  \SetKwFunction{Extract}{Extract}

{\CompProb}$(\mathcal M,s_0,A,\theta)$\;

$AlmostLose_0 \leftarrow \emptyset$;  $AlmostLose_1 \leftarrow \emptyset$; $S'\leftarrow \emptyset$\;
$A'\leftarrow \emptyset$; $\Front \leftarrow \emptyset$; $\Insert(\Front,s_0)$\;
 \While{$\Front\neq \emptyset$}
 { 
    $s \leftarrow \Extract(\Front)$;
      $S'\leftarrow S' \cup \{s\}$\; 
       \lIf {$f_0(s)\leq \frac \theta 2$}{$AlmostLose_0\leftarrow AlmostLose_0 \cup \{s\}$}  
       \lElseIf {$f_1(s)\leq \frac \theta 2$}{$AlmostLose_1\leftarrow AlmostLose_1 \cup \{s\}$}  
       \lElseIf{$s\in A$}  {$A'\leftarrow A' \cup \{s\}$}
       \lElse 
       {\lFor {$s\rightarrow s' \wedge s' \notin S'$} {$\Insert(\Front,s')$}}
       \vspace*{-0.5cm}
 }   
  \lIf {$A'=\emptyset$} {\Return$(0,\theta)$}
  $Abs\leftarrow AlmostLose_0 \cup AlmostLose_1\cup  A'$\;
  \lFor {$s\in Abs$} {$p'(s,s)\leftarrow 1$}
  \lFor {$s\in S' \setminus Abs \wedge s'\in S'$} {$p'(s,s')\leftarrow p(s,s')$}
 $preach \leftarrow {\CompFinProb}(\mathcal M',s_0)$\;
\Return{\scalebox{0.85}{$(preach(A'),preach(A')+preach(AlmostLose_0)+\frac \theta 2\cdot preach(AlmostLose_1))$}}

\caption{Bounding the reachability probability}
 \label{algo:prob-reach-divbis}
 \end{algorithm2e}

Let us describe the first algorithm that does not require the decidability of reachability problem. 
It performs an exploration of reachable states from $s_0$ maintaining $S'$, the set
of visited states, and stopping an exploration when 
the current state $s$ fulfills: 
either (1) for some $i\in \{0,1\}$, $f_i(s)\leq \frac \theta 2$ in which case $s$ is inserted in the $AlmostLose_i$ set (initially empty), 
or (2) $s\in A$
in which case $s$ is inserted in $A'$ (initially empty).
When the exploration is ended, if $A'$ is empty, the algorithm returns the interval
$[0,\theta]$. Otherwise it builds $\mathcal M'=(S',\suiv')$ a finite Markov chain over $S'$ whose transition probabilities
are the ones of $\mathcal M$ except for the states of  $AlmostLose_0 \cup AlmostLose_1\cup  A'$, 
 which are made absorbing.
Finally it computes the vector of reachability probabilities starting from $s_0$ in $\mathcal M'$ (function \CompFinProb) and returns
 the interval $[preach(A'),preach(A')+preach(AlmostLose_0)+\frac \theta 2\cdot preach(AlmostLose_1)]$.
 The next proposition establishes the correctness of the algorithm.

\begin{restatable}{proposition}{algocorrectdivbis}
\label{prop:prob-reach-divbis} 
Let  $\mathcal M$ be a divergent Markov chain with $s_0\in S$, $A\subseteq S$ and $\theta>0$.
Then Algorithm~\ref{algo:prob-reach-divbis} solves the CRP problem. 
\end{restatable}
\begin{proof}
Let us show that Algorithm~\ref{algo:prob-reach-divbis} terminates. Any state $s$ inserted in $S'$ fulfills $\min(f_0(s),f_1(s))>\frac \theta 2$ or
is reachable by a transition from such a state. Due to the third condition of divergence, there is a finite number of
such states,  the termination of Algorithm~\ref{algo:prob-reach-divbis} is ensured.

\noindent
Every state $s$ inserted in $AlmostLose_0$ fulfill $f_0(s)\leq \frac \theta 2$. Due to the first item
of divergence, $\prob_{\mathcal M,s_0}(\mathbf{F} AlmostLose_0)\leq \frac \theta 2$.
Since a path in $\mathcal M'$ from $s_0$ that reaches $AlmostLose_0$
is also a path in $\mathcal M$, $\prob_{\mathcal M',s_0}(\mathbf{F} AlmostLose_0)\leq \prob_{\mathcal M,s_0}(\mathbf{F} AlmostLose_0)$.
Thus the interval returned by the algorithm has length at most $\theta$. It remains to prove that 
$\prob_{\mathcal M,s_0}({\bf F} A)$ belongs to this interval.

\noindent
Since a path in $\mathcal M'$ from $s_0$ that reaches $A'\subseteq A$ is also a path in $\mathcal M$, 
$preach(A')\leq \prob_{\mathcal M,s_0}({\bf F} A)$.
\noindent
Consider in $\mathcal M$ a path starting from $s_0$ that reaches $A$. There are three possible (exclusive) cases:
\begin{itemize}
	\item either it is a path in $\mathcal M'$;
	\item or the last state that it reaches in $\mathcal M'$ is a state $s$ fulfilling $f_0(s)\leq \frac \theta 2$. The probability
	of such paths is thus bounded by $\prob_{\mathcal M',s_0}(\mathbf{F} AlmostLose_0)$;
	\item or the last state that it reaches in $\mathcal M'$ is a state $s$ fulfilling $f_0(s)>\frac \theta 2$ and $f_1(s)\leq \frac \theta 2$. For any such state
	the probability to reach $A$ is at most $\frac \theta 2$. The cumulated probability of such paths
	is bounded by $\prob_{\mathcal M',s_0}(\mathbf{F} AlmostLose_1)\cdot \frac \theta 2$.	
\end{itemize}
Thus the upper bound of the interval returned by the algorithm is correct. Observe that the interval $[0,\theta]$ returned
when $A'=\emptyset$ may be larger than the one that could have been returned but avoids the computation
of the reachability probabilities in $\mathcal M'$.
\end{proof}

\bigskip

When the reachability problem is decidable, one can modify the previous algorithm such that
when $A$ is reachable from $s_0$, the algorithm returns an interval whose lower bound 
is strictly greater than 0.  
It first decides whether $A$ is reachable from $s_0$. In the positive case, by a breadth first exploration,
it discovers a path from $s_0$ to $A$ and decreases $\theta$ according to the values of $f_0$ and $f_1$ along this path. 
This ensures that the states of this path will belong to $S'$ to be built later on.
The remaining part of the algorithm is similar to Algorithm~\ref{algo:prob-reach-divbis} except that using decidability
of the reachability problem, we insert in $Lose$  
the encountered states from which one cannot reach $A$.
These states will also be made absorbing in $\mathcal M'$.  

 \begin{algorithm2e}[!h]
  \LinesNumbered
  \DontPrintSemicolon
  \SetKwFunction{CompProb}{CompProb}
  \SetKwFunction{CompFinProb}{CompFinProb}
 \SetKwFunction{PosReach}{PosReach}
  \SetKwFunction{Push}{Push}
  \SetKwFunction{Pop}{Pop}
  \SetKwFunction{Insert}{Insert}
  \SetKwFunction{Extract}{Extract}

{\CompProb}$(\mathcal M,s_0,A,\theta)$\; 
\lIf{{\bf not} $s_0\rightarrow^* A$}{\Return $(0,0)$}
$\Insert(Queue,s_0)$; $S'\leftarrow \emptyset$\; 
$cont \leftarrow {\bf true}$; $pre[s_0]\leftarrow s_0$\;
 \While{$cont$}
 { 
    $s \leftarrow \Extract(Queue)$;  $S'\leftarrow S' \cup \{s\}$\; 
       \lIf {$s\in A$}{$cont \leftarrow {\bf false}$}
       \ElseIf{$s\rightarrow^* A$}  
       {
          \lFor {$s\rightarrow s' \wedge s' \notin S'$} {$\Insert(Queue,s')$; $pre[s']\leftarrow s$}
       }
 } 
$\theta\leftarrow \min(\theta,f_0(s),f_1(s))$\; 
\lWhile{$s\neq pre[s]$}{ $s\leftarrow pre[s]$; $\theta \leftarrow \min(\theta,f_0(s),f_1(s))$}
 $S'\leftarrow \emptyset$;
$AlmostLose_0 \leftarrow \emptyset$;
$AlmostLose_1 \leftarrow \emptyset$\; 
$Lose \leftarrow \emptyset$; $\Front \leftarrow \emptyset$; $\Insert(\Front,s_0)$\;
 \While{$\Front\neq \emptyset$}
 { 
    $s \leftarrow \Extract(\Front)$;
      $S'\leftarrow S' \cup \{s\}$\; 
      \uIf{{\bf not} $s\rightarrow^* A$}{$Lose\leftarrow Lose \cup \{s\}$}
      \uElseIf {$f_0(s)\leq \frac \theta 2$}{$AlmostLose_0\leftarrow AlmostLose_0 \cup \{s\}$}  
       \uElseIf {$f_1(s)\leq \frac \theta 2$}{$AlmostLose_1\leftarrow AlmostLose_1 \cup \{s\}$}  
       \uElseIf{$s\in A$}  {$A'\leftarrow A' \cup \{s\}$}
       \Else 
       {
          \lFor {$s\rightarrow s' \wedge s' \notin S'$} {$\Insert(\Front,s')$}
       }
 }   
 {\small $Abs\leftarrow Lose \cup AlmostLose_0 \cup AlmostLose_1\cup  A'$}\;
  \lFor {$s\in Abs$} {$p'(s,s)\leftarrow 1$}
  \lFor {$s\in S' \setminus Abs \wedge s'\in S'$} {$p'(s,s')\leftarrow p(s,s')$}
 $preach \leftarrow {\CompFinProb}(\mathcal M',s_0)$\;
 {\small \Return$(preach(A'),preach(A')+preach(AlmostLose_0)+\frac \theta 2\cdot preach(AlmostLose_1))$}

\caption{Bounding the reachability probability when reachability is decidable}
 \label{algo:prob-reach-divter}
 \end{algorithm2e}

\begin{restatable}{proposition}{algocorrection} 
  \label{prop:prob-reach-div2}
Algorithm~\ref{algo:prob-reach-divter} terminates and computes the interval 
$[0,0]$ if $A$ is unreachable from $s_0$ and otherwise an interval of length at most $\theta$, not containing 0
and containing $\prob_{\mathcal M,s_0}({\bf F} A)$. 
\end{restatable}

\begin{proof}
The proof of correctness of Algorithm~\ref{algo:prob-reach-divter} is very similar to the one of
Algorithm~\ref{algo:prob-reach-divbis}. There are only two observations to be done. First a path in $\mathcal M$ 
that reaches a state of $Lose$ without visiting before $A$ cannot reach afterwards $A$ 
which justifies the fact that they are absorbing in $\mathcal M'$ and do not occur in the computation
of the interval. 

\smallskip\noindent
The additional preliminary stage when $A$ is reachable from $s_0$ consists in (1) finding a reachability path from $s_0$ to $A$
and then decreasing $\theta$ in such a way that no state along this path will be discarded during the main exploration.
Since the path discovered during the first stage belong to $\mathcal M'$,
the (non null) probability of this path will lower bound the lower bound of the interval returned by the algorithm. 
\end{proof}

\section{Divergence of random walks}
\label{sec:rw}

We now study the decidability of the divergence property
in random walks with dynamic weights, i.e., a process 
moves one step left or right on an integer line with  
probabilities determined by its current position. In what follows,  
the divergence problem consists in checking the existence of 
$f_0$ and $f_1$ such that the divergence property holds in a given model.

\begin{restatable}{theorem}{ppdaundec} 
The divergence problem of random walks with dynamic weights is undecidable. 
\end{restatable}
\begin{proof}
We will reduce Hilbert's tenth problem to the divergence problem. 
Let us briefly describe Hilbert's tenth problem, dynamic, known to be undecidable. 
Let $P\in \integer[X_1,\ldots X_k]$ be an integer polynomial with $k$ variables. This problem
asks whether there exist $n_1,\ldots, n_k\in \nat$ such that $P(n_1,\ldots,n_k)=0$.

\smallskip\noindent
We now define dynamic weights above the structure of
$\mathcal M_1$ of Figure~\ref{fig:rw} as follows. 
\begin{itemize}
	\item $W(n, n+1)=\min(P^2(n_1,\ldots,n_k)+1 \mid n_1+\ldots +n_k\leq n) $;
	\item $W(n, n-1)=1$;
	\item  $W(0, 1)=1$.
\end{itemize}

\noindent
This function is obviously computable. 
One studies the divergence w.r.t.  $s_0=1$ and $A=\{0\}$. 
By construction, $p_n=\frac {W(n, n+1)} {W(n, n-1)+W(n, n+1)}$ in  $\mathcal M_1$.  
 
\noindent
 $\bullet$ Assume there exist  $n_1,\ldots, n_k\!\in\! \nat$ s.t. $P(n_1,\ldots,n_k)\!=\!0$.
Then for $n\geq \sum_{i} n_i$, $p_n=\frac 1 2$ implying the recurrence of the Markov chain.
Assume there exist functions $f_0$, $f_1$ fulfilling the first two items of Definition~\ref{definition:divergent}, then these functions
are necessarily the constant function 1 which falsifies the last item of the definition for any $\theta<1$.
Thus  $\mathcal M_{1}$ is not divergent.

\noindent
$\bullet$
Assume there do not exist  $n_1,\ldots, n_k\in \nat$ such that $P(n_1,\ldots,n_k)=0$. Thus for all $n$, $p_n\geq \frac 2 3$
implying $\rho_n\leq \frac 1 2$.
Then $\mathcal M_{1}$ fulfills the hypotheses of Proposition~\ref{proposition:f01}
with function $f_1$ defined by $f_1(n)=\frac 1 {2^n}$.
\end{proof}

\bigskip

Due to this negative result on such a basic model, it is clear that one must restrict the possible weight functions. 
A random walk is said polynomial if for each state $n$, the weight 
$W(n,n+1)$ and $W(n,n-1)$  are positive integer polynomials (i.e. whose coefficients are non-negative and the constant one is positive)
whose single variable is $n$.

\begin{restatable}{theorem}{decdivpda} 
\label{dec-div-pPDA}
The divergence problem for  polynomial random walks is decidable  (in linear time). 
\end{restatable}
\begin{proof}
Let us assume that $s_0=n_\iota$ and $A=\{n_f\}$ with $n_\iota>n_f$.
The other cases either reduce to this one or do not present difficulties.

As for the previous theorem, we consider also the Markov chain $\mathcal M_1$ of Figure~\ref{fig:rw}, with 
dynamic weights.
Precisely, $W(n, n+1)$ and  $W(n, n-1)$ are polynomial depending on the value of $n$, for all $n>0$.
While $W(0,1)$ is irrelevant since it is the single transition outgoing from the state $0$. 
We introduce $\rho_n=\frac{W(n, n-1)}{W(n, n+1)}$ for $n>0$
and we will show how to decide whether this Markov chain is recurrent.

\noindent
Let us write $W(n, n-1)=\sum_{i\leq d}a_in^i$ and  $W(n, n+1)=\sum_{i\leq d'}a'_in^i$
with coefficients in $\nat$.
We perform a case analysis.  

\noindent $\bullet$ When:
\begin{itemize} 
	\item $d'<d$
	\item or $d'=d$,  $i_0=\max(i \mid a_i\neq a'_i)$ exists and $a_{i_0}>a'_{i_0}$
	\item or $W(n, n+1)= W(n, n-1)$
\end{itemize}
Then there exists $n_0$ such that for all $n \geq n_0$, \\
$W(n, n+1)\leq W(n, n-1)$ implying $\rho_n\geq 1$.\\
Thus for all $n\geq n_0$, $\prod_{1\leq m\leq n}\rho_m\geq \prod_{1\leq m\leq n_0}\rho_m$ 
implying $\sum_{n\in \nat}\prod_{1\leq m\leq n}\rho_m=\infty$ yielding recurrence.

\smallskip\noindent
$\bullet$ {\bf Case} $d'=d$ and $i_0=\max(i \mid a_i\neq a'_i)$ exists and $a_{i_0}<a'_{i_0}$ and
$i_0\leq d-2$. Then there exists $n_0$ and $\alpha>0$ such that for all $n>n_0$,
$\rho_n\geq 1- \frac{\alpha}{n^2}$. 
Observe that:
\begin{align*}
 \prod_{n_0<m\leq n} \rho_m &\geq&  \prod_{n_0<m\leq n} 1- \frac{\alpha}{m^2}&\geq& \prod_{n_0<m\leq n} e^{- \frac{2\alpha}{m^2}}\\
&=& e^{- \sum_{n_0<m\leq n} \frac{2\alpha}{m^2}}
&\geq& e^{- \sum_{n_0<m} \frac{2\alpha}{m^2}}>0
\end{align*}
Thus $\sum_{n\in \nat}\prod_{1\leq m\leq n}\rho_m=\infty$ yielding recurrence.

\smallskip\noindent
$\bullet$ {\bf Case} $d'=d$ and $i_0=\max(i \mid a_i\neq a'_i)$ exists, $i_0= d-1$ and $0<\frac{a'_{d-1}-a_{d-1}}{a_d}\leq 1$. 
Let $\alpha=\frac{a'_{d-1}-a_{d-1}}{a_d}$.\\ 
Then there exists $n_0$ and $\beta>0$ such that for all $n>n_0$,\\
$\rho_n\geq 1- \frac{\alpha}{n}-\frac{\beta}{n^2}$. 
Observe that:\\
\begin{align*} 
\prod_{n_0<m\leq n} \rho_n&\geq\  \prod_{n_0<m\leq n} 1- \frac{\alpha}{m}- \frac{\beta}{m^2}\\
&\geq\ \prod_{n_0<m\leq n} e^{- \frac{\alpha}{m}- \frac{\beta}{m^2}-(\frac{\alpha}{m}- \frac{\beta}{m^2})^2}\\
&\geq \prod_{n_0<m\leq n} e^{- \frac{\alpha}{m}- \frac{\beta'}{m^2}}\\
\mbox{\emph{for some }}\beta'& \\
&=\ e^{-\sum_{n_0<m\leq n}\frac{\alpha}{m}+ \frac{\beta'}{m^2}}\\
&\geq\ e^{-\sum_{n_0<m} \frac{\beta'}{m^2}}e^{- \alpha\log(n)}=\frac{e^{-\sum_{n_0<m} \frac{\beta'}{m^2}}}{n^\alpha}
\end{align*} 
Thus $\sum_{n\in \nat}\prod_{1\leq m\leq n}\rho_m=\infty$ yielding recurrence.

\noindent
$\bullet$ {\bf Case} $d'=d$ and $i_0=\max(i \mid a_i\neq a'_i)$ exists, $i_0= d-1$ and $\frac{a'_{d-1}-a_{d-1}}{a_d}> 1$. 
Let $\alpha=\frac{a'_{d-1}-a_{d-1}}{a_d}$.\\ 
Then there exists $n_0$ and $1<\alpha'<\alpha$ such that
for all $n\geq n_0$,  $\rho_n\leq  1- \frac{\alpha'}{n}$.
Observe that:
\begin{align*}
 \prod_{n_0<m\leq n} \rho_m&\leq\  \prod_{n_0<m\leq n} 1- \frac{\alpha'}{m}\\
 &\leq\ \prod_{n_0<m\leq n} e^{- \frac{\alpha'}{m}}\\
&=\ e^{-\sum_{n_0<m\leq n}\frac{\alpha'}{m}}\\
&\leq\ e^{- \alpha'(\log(n+1)-\log(n_0+1))}\\
&=\ \frac{e^{ \alpha'\log(n_0+1)}}{(n+1)^{\alpha'}}
\end{align*}
Thus $\sum_{n\in \nat}\prod_{1\leq m\leq n}\rho_m<\infty$ implying non recurrence.

\noindent
$\bullet$  When: 
\begin{itemize} 
	\item $d=d'$ and $a'_d>a_d$ or
	\item $d<d'$. 
\end{itemize}
Then there exists $n_0$ and $\alpha<1$ such that
for all $n\geq n_0$,  $\rho_n \leq \alpha$.
Thus $\sum_{n\in \nat}\prod_{1\leq m\leq n}\rho_m<\infty$ implying non recurrence.

\smallskip\noindent
This concludes the proof that the recurrence is decidable for this model. 
Let us focus on divergence.

\noindent
$\bullet$ Assume that  this random walk Markov chain 
is  recurrent. Then the probability to reach any state from $n_\iota$ is 1 and the probability
to reach from any state $n_f$ is 1.
If there are functions $f_0$, $f_1$ fulfilling the first two items of Definition~\ref{definition:divergent}, then these functions
are necessarily the constant function 1, thus falsifying the last item of the definition for any $\theta<1$.
Thus  this Markov chain is not divergent.

\noindent
$\bullet$
Assume that  the Markov chain
is transient. The probability $preach(n)$
starting from $n$ with $n>n_f$ to reach $n_f$ is given by this formula.
$$preach(n)=\frac{\sum_{m\geq n} \prod_{n_f<k\leq m}\rho_k}{\sum_{m\geq n_f} \prod_{n_f<k\leq m}\rho_k}$$

\noindent
Then an upper bound of $preach(n)$ which corresponds to function $f$ of Proposition~\ref{proposition:f01}, i.e., f is $f_0$ when $f_1=1$ or $f_1$ when $f_0=1$,  is defined by :
\begin{itemize}
	\item when ($d=d'$ and $a'_d>a_d$) or $d<d'$,
	$\sum_{m\geq n} \alpha^{m-n_f}=\frac{\alpha^{n-n_f}}{1-\alpha}$;
	\item when $d'=d$ and $i_0=\max(i \mid a_i\neq a'_i)$ exists,\\ $i_0= d-1$ and $\frac{a'_{d-1}-a_{d-1}}{a_d}> 1$,
	$e^{ \alpha'\log(n_0+1)}\sum_{m\geq n} \frac{1}{(m+1)^{\alpha'}}\leq \frac{e^{ \alpha'\log(n_0+1)}}{(1-\alpha')n^{\alpha'-1}}$.
\end{itemize}
So the Markov chain is divergent.

\noindent
Since recurrence of this random walk Markov chain is decidable, divergence is also decidable. 
\end{proof}
%
 
\section{Divergence of probabilistic Petri nets}
\label{sec:pPN}

We now study the decidability of the divergence property for probabilistic versions of Petri nets (pPN). 
A probabilistic Petri net (resp. a probabilistic VASS) is a Petri net (resp. a VASS) with a computable weight function $W$. 
In previous works~\cite{AbdullaHM07,DBLP:conf/lics/BrazdilCK0VZ18}, 
the weight function $W$ is a \emph{static} one: i.e., a function from $T$ to $\nat_{>0}$. 
As above, we call these models \emph{static} probabilistic Petri nets.
We introduce here a more powerful function where the weight of a transition depends on the current marking.

\begin{definition} A \emph{(dynamic-)probabilistic Petri net} (pPN)\\ $\mathcal N =(P,T, \mathbf{Pre}, \mathbf{Post}, W, \mathbf{m}_0)$
  is defined by:
\begin{itemize}
	\item $P$, a finite set of places; 
	\item $T$, a finite set of transitions; 
	\item $ \mathbf{Pre}, \mathbf{Post} \in \nat^{P\times T}$, resp.
	the pre and post condition matrices; 
	\item $W$, a computable function from $T \times \nat^P$  to $\rat_{>0}$ the weight function;
	\item $\mathbf{m}_0\in \nat^P$, the initial marking. 
\end{itemize}
 \end{definition}
When for all $t\in T$, $W(t,-)$ is a positive polynomial whose variables are the place markings,
we say that  $\mathcal N$ is a \emph{polynomial} pPN.

A marking $\mathbf{m}$ is an element of $\nat^P$. Let $t$ be a transition. Then $t$
is \emph{enabled} in $\mathbf{m}$ if for all $p\in P$,   $\mathbf{m} (p)\geq \mathbf{Pre}(p,t)$.
When enabled, the \emph{firing} of $t$ leads to marking $\mathbf{m}'$ defined for all $p\in P$ by
$\mathbf{m}'(p)=\mathbf{m}(p) +  \mathbf{Post}(p,t)- \mathbf{Pre}(p,t)$
which is denoted by  $\mathbf{m}\xrightarrow{t}\mathbf{m}'$. Let $\sigma=t_1\ldots t_n$
be a sequence of transitions. We define the enabling and the firing of $\sigma$ by induction.
The empty sequence is always enabled in   $\mathbf{m}$ and its firing leads to $\mathbf{m}$.
When $n>0$,   $\sigma$ is enabled if $\mathbf{m}\xrightarrow{t_1}\mathbf{m}_1$
and $t_2  \ldots t_n$ is enabled in $\mathbf{m}_1$. The firing of $\sigma$ leads to the marking
reached by $t_2  \ldots t_n$ from $\mathbf{m}_1$. A marking $\mathbf{m}$ is reachable from $\mathbf{m}_0$
if there is a firing sequence $\sigma$ that reaches $\mathbf{m}$ from $\mathbf{m}_0$. 
 \begin{definition}Let $\mathcal N$ be a pPN. Then the \emph{Markov chain} $\mathcal M_{\mathcal N}=(S_\mathcal N,\suiv_\mathcal N)$ associated with $\mathcal N$ is defined by:
 \begin{itemize}
	\item  $S_\mathcal N$ is the set of reachable markings from $\mathbf{m}_0$;
	\item  Let $\mathbf{m}\in S_\mathcal N$ and $T_{\mathbf{m}}$ be the set of transitions enabled
	in $\mathbf{m}$. If $T_{\mathbf{m}}=\emptyset$ then $\suiv_\mathcal N(\mathbf{m},\mathbf{m})=1$. Otherwise let  
	$W(\mathbf{m})=\sum_{\mathbf{m}\xrightarrow{t}\mathbf{m}_t}	W(t,\mathbf{m})$. 
	Then for all $\mathbf{m}\xrightarrow{t}\mathbf{m}_t$, $\suiv_\mathcal N(\mathbf{m},\mathbf{m}_t)=\frac{W(t,\mathbf{m})}{W(\mathbf{m})}$. 
\end{itemize}
 \end{definition}

We establish below that for polynomial pPNs, divergence is undecidable. We will proceed using a reduction of the following undecidable problem.

Let us recall that a two-counter (Minsky) machine $\mathcal C$ is defined by a set of two counters $\{c_1,c_2\}$ and a set of $n+1$ instructions labelled by $\{0,\ldots ,n\}$, 
where for all $i<n$, the instruction $i$ is of type
\begin{itemize}
	\item either (1)  $c_j \leftarrow c_j+1; \mathbf{ goto~} i'$ 
	with $j\in \{1,2\}$ and $0\leq i' \leq n$
	\item or (2) $\mathbf{if~} c_j>0 \mathbf{~then~} c_j \leftarrow c_j-1; \mathbf{goto~} i'$ $\mathbf{else\ goto~} i''$
	with $j\in \{1,2\}$ and $0\leq i',i'' \leq n$ 
\end{itemize}
and the instruction
$n$ is $\mathbf{halt}$. The program machine starts at instruction $0$ and halts if it reaches instruction $n$.

The halting problem for  two-counter machines asks, given  a two-counter machine $\mathcal C$ and initial values of counters,  
whether $\mathcal C$ eventually halts. It is undecidable~\cite{Minsky67}.

We introduce a subclass of  two-counter machines that we call \emph{normalized}.
A normalized two-counter machine $\mathcal C$ starts by resetting its counters and on termination resets its counters before halting.  

\smallskip
In a normalized two-counter machine $\mathcal C$, given any initial values $v_1,v_2$,
 $\mathcal C$ halts with $v_1,v_2$ if and only if  $\mathcal C$ halts with initial values $0,0$.
 Moreover when $\mathcal C$ halts, the values of the counters are null.
  
\noindent
{\bf Normalized two counters machine.}
The two
first instructions of a normalized two counters machine are:
\begin{itemize}
	\item$0: \mathbf{~if~} c_1>0 \mathbf{~then~} c_1 \leftarrow c_1-1; \mathbf{goto~} 0$ $\mathbf{else\ goto~} 1$
	\item$1: \mathbf{~if~} c_2>0 \mathbf{~then~} c_2 \leftarrow c_2-1; \mathbf{goto~} 1$ $\mathbf{else\ goto~} 2$
\end{itemize}
The three
last instructions of a normalized CM are:
\begin{itemize}
	\item$n\!-\!2: \mathbf{~if~} c_1>0 \mathbf{~then~} c_1 \leftarrow c_1-1; \mathbf{goto~} n\!-\!2$ $\mathbf{else\ goto~} n\!-\!1$
	\item$n\!-\!1: \mathbf{~if~} c_2>0 \mathbf{~then~} c_2 \leftarrow c_2-1; \mathbf{goto~} n\!-\!1$ $\mathbf{else\ goto~} n$
	\item$n:\mathbf{halt}$
\end{itemize}
For $2<i<n-2$, the labels occurring in instruction $i$ belong to $\{0,\ldots,n-2\}$.

The halting problem for two-counter machines can be reduced to the halting problem  
for normalized two-counter machines, which is thus undecidable.

\begin{lemma}
The halting problem of normalized counter machines is undecidable.
\end{lemma}

\begin{proof}
Let $\mathcal C$ be a two-counter machine with initial values $v_1,v_2$,
 one builds the normalized two-counter machine $\mathcal C_{v_1,v_2}$ by adding after the two first instructions
 of a normalized  two-counter machine, $v_1$ incrementations of $c_1$ followed by $v_2$ incrementations of $c_2$
 followed by  the instructions $\mathcal C$ where the halting instruction has been replaced by  the last three instructions
 of a normalized  two-counter machine. The normalized  two-counter machine $\mathcal C_{v_1,v_2}$ halts if and only if 
  $\mathcal C$ with initial values $v_1,v_2$ halts.
\end{proof}  

\medskip

The following fact will be used in the next proofs.
Let $X$ be a random variable with range in $\nat$ and $h$ be a strictly decreasing function from $\real_{\geq 0}$ to $\real_{\geq 0}$.
Assume that for all $n\geq n_0$, $\prob(X\geq n)\leq h(n)$.\\ Then for all $\theta$ such that $\theta\leq h(n_0)$, 
$$
\begin{array}{r c l c}
\prob(h(X)\leq \theta )&=&\prob(X\geq h^{-1}(\theta) )&\hspace*{3cm}\\
&\multicolumn{3}{l}{\mbox{(since } h^{-1} \mbox { decreasing entails } a \leq b \Leftrightarrow h^{-1}(a) \geq h^{-1}(b)\mbox{)}}\\
&=&  \prob(X\geq \lceil h^{-1}(\theta) \rceil)&\\
&\multicolumn{3}{l}{\mbox{(since } X\mbox{ is integer-valued)}}\\
&\leq& h(\lceil h^{-1}(\theta) \rceil)&\\
&\multicolumn{3}{l}{\mbox{(by hypothesis on } h\mbox{)}}\\
&\leq& \theta&\\
&\multicolumn{3}{l}{\mbox{(since } h \mbox { decreasing entails } a \geq b \Leftrightarrow h(a) \leq h(b)\mbox{ and } \lceil h^{-1}(\theta)\rceil\geq h^{-1}(\theta) \mbox{)} }
\end{array}
$$

Contrary to the previous result, restricting the weight functions to be polynomials
does not yield decidability for pPNs.
\begin{restatable}{theorem}{pndivundecidable} 
\label{theorem:polynomial-pPN-divergence-undecidable}
The divergence problem of polynomial pPNs  w.r.t. an upward closed set is undecidable. 
\end{restatable}
\begin{proof} 
We reduce the reachability problem of normalized two-counter machines to the  divergence problem of polynomial pPNs. 
Let $\mathcal C$ be a normalized two-counter machine. The  pPN $\mathcal N_\mathcal C$ is built as follows.
Its set of places is $\{p_i \mid 0\leq i\leq n\} \cup \{q_i \mid i \mbox{~is a test instruction}\}   \cup \{c_j, c'_j \mid 1\leq j\leq 2\} \cup \{init,sim,stop\}$.
The initial marking is $\mathbf{m}_0=init$ and the set $A$ for divergence is defined by $A = \{\mathbf{m} \mid \mathbf{m}(p_0)\geq 1\}$.

 \begin{figure}
\begin{center}
\subfloat[]{
\begin{tikzpicture}[xscale=0.7,yscale=0.7]

\path (0,0) node[draw,circle,inner sep=2pt,minimum size=0.6cm,
label={[xshift=0cm, yshift=0cm]$init$}] (init) {};

\path (0,-2) node[draw,circle,inner sep=2pt,minimum size=0.6cm,
label={[xshift=0cm, yshift=0cm]$c'_1$}] (c1prime) {};

\path (0,-4) node[draw,circle,inner sep=2pt,minimum size=0.6cm,
label={[xshift=0cm, yshift=-1.2cm]$c'_{2}$}] (c2prime) {};

\path (2,-4) node[draw,circle,inner sep=2pt,minimum size=0.6cm,
label={[xshift=0cm, yshift=-1.2cm]$p_0$}] (p0) {};

\path (-2,-2) node[draw,rectangle,inner sep=2pt,minimum width=0.4cm,minimum height=0.2cm,
label={[xshift=-0.8cm, yshift=-0.4cm]$inc_{init}$}] (inc) {};

\path (2,-2) node[draw,rectangle,inner sep=2pt,minimum width=0.4cm,minimum height=0.2cm,
label={[xshift=0.8cm, yshift=-0.4cm]$start$}] (start) {};

\draw[arrows=-latex] (init) -- (-2,0) --(inc) ;
\draw[arrows=-latex] (inc) -- (-2,0) --(init) ;

\draw[arrows=-latex] (init) -- (2,0) --(start) ;

\draw[arrows=-latex] (inc)--(c1prime);

\draw[arrows=-latex] (inc)--(-2,-4)--(c2prime);
\draw[arrows=-latex] (start)  --(p0) ;

\end{tikzpicture}
}
\hfill
\subfloat[]{
 
\begin{tikzpicture}[xscale=0.7,yscale=0.7]

\path (0,0) node[draw,circle,inner sep=2pt,minimum size=0.6cm,
label={[xshift=0cm, yshift=0cm]$p_i$}] (pi) {};

\path (0,-2) node[draw,circle,inner sep=2pt,minimum size=0.6cm,
label={[xshift=0cm, yshift=0cm]$c_j$}] (cj) {};

\path (-4,-2) node[draw,circle,inner sep=2pt,minimum size=0.6cm,
label={[xshift=0cm, yshift=0cm]$c'_j$}] (cjprime) {};

\path (0,-4) node[draw,circle,inner sep=2pt,minimum size=0.6cm,
label={[xshift=0cm, yshift=-1.2cm]$p_{i'}$}] (piprime) {};

\path (2,-4) node[draw,circle,inner sep=2pt,minimum size=0.6cm,
label={[xshift=0cm, yshift=-1.2cm]$stop$}] (stop) {};

\path (-2,-2) node[draw,rectangle,inner sep=2pt,minimum width=0.4cm,minimum height=0.2cm,
label={[xshift=-0.4cm, yshift=0cm]$inc_i$}] (inci) {};

\path (2,-2) node[draw,rectangle,inner sep=2pt,minimum width=0.4cm,minimum height=0.2cm,
label={[xshift=0.8cm, yshift=-0.4cm]$exit_i$}] (exiti) {};

\draw[arrows=-latex] (pi) -- (-2,0) --(inci) ;
\draw[arrows=-latex] (pi) -- (2,0) --(exiti) ;

\draw[arrows=-latex] (cjprime)--(inci);
\draw[arrows=-latex] (inci)--(cj);

\draw[arrows=-latex] (inci)--(-2,-4)--(piprime);
\draw[arrows=-latex] (exiti)  --(stop) ;

\end{tikzpicture}
}
\end{center}
\caption{Initialization stage (a); \ \ \ \ \ \ $i: c_j \leftarrow c_j+1; \mathbf{ goto~} i'$ (b)}
\label{fig:incrementbis}

\end{figure}

\begin{figure}
\begin{center}
\subfloat[]{
\begin{tikzpicture}[xscale=0.7,yscale=0.7]

\path (0,0) node[draw,circle,inner sep=2pt,minimum size=0.6cm,
label={[xshift=0cm, yshift=0cm]$p_i$}] (pi) {};

\path (0,-2) node[draw,circle,inner sep=2pt,minimum size=0.6cm,
label={[xshift=0cm, yshift=0cm]$c_j$}] (cj) {};

\path (-4,-2) node[draw,circle,inner sep=2pt,minimum size=0.6cm,
label={[xshift=0cm, yshift=0cm]$c'_j$}] (cjprime) {};

\path (2,-2) node[draw,circle,inner sep=2pt,minimum size=0.6cm,
label={[xshift=0.6cm, yshift=-0.6cm]$q_i$}] (qi) {};

\path (-2,-4) node[draw,circle,inner sep=2pt,minimum size=0.6cm,
label={[xshift=0cm, yshift=-1.2cm]$p_{i'}$}] (piprime) {};

\path (2,-4) node[draw,circle,inner sep=2pt,minimum size=0.6cm,
label={[xshift=0cm, yshift=-1.2cm]$p_{i''}$}] (pisecond) {};

\path (4,-4) node[draw,circle,inner sep=2pt,minimum size=0.6cm,
label={[xshift=0cm, yshift=-1.2cm]$stop$}] (stop) {};

\path (0,-5.5) node[draw,circle,inner sep=2pt,minimum size=0.6cm,
label={[xshift=-0.8cm, yshift=-0.5cm]$sim$}] (sim) {};

\path (-2,-2) node[draw,rectangle,inner sep=2pt,minimum width=0.4cm,minimum height=0.2cm,
label={[xshift=-0.4cm, yshift=0cm]$dec_i$}] (deci) {};

\path (2,-1) node[draw,rectangle,inner sep=2pt,minimum width=0.4cm,minimum height=0.2cm,
label={[xshift=0.8cm, yshift=-0.4cm]$begZ_i$}] (begZi) {};

\path (2,-3) node[draw,rectangle,inner sep=2pt,minimum width=0.4cm,minimum height=0.2cm,
label={[xshift=0.8cm, yshift=-0.4cm]$endZ_i$}] (endZi) {};

\path (4,-2) node[draw,rectangle,inner sep=2pt,minimum width=0.4cm,minimum height=0.2cm,
label={[xshift=0.8cm, yshift=-0.4cm]$exit_i$}] (exiti) {};

\path (0,-4) node[draw,rectangle,inner sep=2pt,minimum width=0.4cm,minimum height=0.2cm,
label={[xshift=-0.6cm, yshift=-0.4cm]$rm_i$}] (rmi) {};

\draw[arrows=-latex] (pi) -- (-2,0) --(deci) ;
\draw[arrows=-latex] (pi) -- (2,0) --(begZi) ;
\draw[arrows=-latex] (begZi) --(qi) ;
\draw[arrows=-latex] (qi) -- (endZi);
\draw[arrows=-latex] (cj)--(deci);
\draw[arrows=-latex] (deci)--(cjprime);
\draw[arrows=-latex] (endZi)--(pisecond);
\draw[arrows=-latex] (deci)--(piprime);

\draw[arrows=-latex] (sim) --(rmi) node[pos=0.5,right] {$2$}  ;
\draw[arrows=-latex] (exiti)  --(stop) ;

\draw[arrows=-latex] (pi)  --(4,0)--(exiti) ;
\draw[arrows=-latex] (qi)  --(rmi) ;
\draw[arrows=-latex] (rmi)  --(qi) ;
\draw[arrows=-latex] (cj)  --(rmi) ;
\draw[arrows=-latex] (rmi)  --(cj) ;

\end{tikzpicture}
}
\hfill
\subfloat[]{
\begin{tikzpicture}[xscale=0.7,yscale=0.7]

\path (0,0) node[draw,circle,inner sep=2pt,minimum size=0.6cm,
label={[xshift=0cm, yshift=0cm]$p_n$}] (pn) {};

\path (0,-2) node[draw,circle,inner sep=2pt,minimum size=0.6cm,
label={[xshift=0cm, yshift=0cm]$sim$}] (sim) {};

\path (0,-4) node[draw,circle,inner sep=2pt,minimum size=0.6cm,
label={[xshift=0cm, yshift=-1.2cm]$p_0$}] (p0) {};

\path (-2,-2) node[draw,rectangle,inner sep=2pt,minimum width=0.4cm,minimum height=0.2cm,
label={[xshift=-0.8cm, yshift=-0.4cm]$again$}] (again) {};

\draw[arrows=-latex] (pn) -- (-2,0) --(again) ;
\draw[arrows=-latex] (again)--(sim);
\draw[arrows=-latex] (again)--(-2,-4)--(p0);

\end{tikzpicture}
}
\end{center}
\caption{\begin{small}$i:\mathbf{if~} c_j>0 \mathbf{~then~} c_j \leftarrow c_j-1; \mathbf{goto~} i' \mathbf{else\ goto~} i''$\end{small} (a); \ \ \ halt instruction (b)}
\label{fig:haltbis}
\end{figure}

\smallskip\noindent
Places $c'_1$ and $c'_2$ are ``complementary places" for the counters. 
During the initialization stage, transition  $inc_{init}$ repeatedly  increments
these places until transition $start$ is fired unmarking $init$. The weight of  $inc_{init}$ 
and $start$ is 1.

\smallskip\noindent
The set of transitions is instantiated from a pattern per type of instruction. The initialization stage is performed by 
the subnet of Figure~\ref{fig:incrementbis} (a). The pattern for the
incrementation instruction is depicted in Figure~\ref{fig:incrementbis} (b). The pattern for the test
instruction is depicted in Figure~\ref{fig:haltbis} (a). The pattern for the halt instruction is
depicted in Figure~\ref{fig:haltbis} (b). 

\smallskip\noindent
Before specifying the weight function $W$, let us describe the qualitative behaviour of this net.
$\mathcal N_\mathcal C$ performs repeatedly a \emph{weak} simulation of  $\mathcal C$. As usual
 since the zero test does not exist in Petri nets, during a test instruction $i$ the simulation can follow the zero branch while the corresponding counter
is non null (transitions  $begZ_i$ and $endZ_i$). If the net has cheated then with transition $rm_i$, it can remove
tokens from $sim$ (two tokens each time). In addition when the instruction is not $\mathbf{halt}$, instead of simulating it, it can \emph{exit}
the simulation by putting a token in $stop$. 
The simulation of the $\mathbf{halt}$ instruction consists in restarting the simulation and incrementing the simulation counter
$sim$.

 \smallskip\noindent
 For $j\in \{1,2\}$, the sum of markings of $c_j$ and $c'_j$ remains constant
 fixed by the number of firings  $inc_{init}$ during the initialization stage.
 Thus the simulation of the $i^{th}$ instruction when $i$ increments  $c_j$ cannot proceed
 if $c_j$ has reached this value letting only fireable transition $exit_i$.

\smallskip\noindent
Thus the set of reachable markings is included in the following set of markings 
$\{p_i +xc_1+yc_2+(s-x)c'_1+(s-y)c'_2+z sim\mid 0\leq i\leq n, x,y,s,z \in \nat\} \cup \{q_i +xc_1+yc_2+(s-x)c'_1+(s-y)c'_2+z sim\mid  i \mbox{~is a test instruction}, x,y,s,z \in \nat\} \cup \{stop +xc_1+yc_2+(s-x)c'_1+(s-y)c'_2+z sim\mid x,y,s,z \in \nat\}$. 

\smallskip\noindent
Let us specify the weight function. For all incrementation instruction $i$, 
$W(inc_i,\mathbf{m})=\mathbf{m}(sim)^2+1$.
For all test instruction $i$, 
$W(begZ_i,\mathbf{m})=\mathbf{m}(sim)^2+1$,
$W(dec_i,\mathbf{m})=2\mathbf{m}(sim)^4+2$
and $W(rm_i,\mathbf{m})=2$.
All other weights are equal to 1.

\smallskip\noindent
Thus after the initialization stage for all $j\in \{1,2\}$ the probability that a reachable marking
 $\mathbf{m}$ fufills $\mathbf{m}(c_j)+\mathbf{m}(c'_j)\geq n$ is exactly $2^{-n}$.

 \smallskip\noindent
 We will establish that $\mathcal N_\mathcal C$ is divergent w.r.t. to $\mathbf{m}_0$ and $A$ if and only if $\mathcal C$ does not halt.

\smallskip\noindent
$\bullet$ Assume that $\mathcal C$ halts and consider its execution $\rho$ with initial values $(0,0)$.
Let $\ell=|\rho|$ the length of this execution. Consider now $\sigma$ the infinite sequence of $\mathcal N_\mathcal C$
that first fires  $\ell$ times $inc_{init}$, then fires $start$ and  infinitely performs the correct simulation of this execution
(possible due to the initialization stage). 
\smallskip\noindent
After every simulation of $\rho$, the marking of $sim$ is incremented and it is never decremented
since (due to the correctness of the simulation) every time a transition $begZ_i$ is fired, the corresponding counter place
$c_j$ is unmarked which forbids the firing of $rm_i$. So during the $(n+1)^{th}$ simulation of $\rho$,
the marking of $sim$ is equal to $n$.

\smallskip\noindent
So consider the  probability of the correct simulation of an instruction $i$  during the $(n+1)^{th}$ simulation.
\begin{itemize}
	\item If $i$ is an incrementation then the weight of $inc_i$ is $n^2+1$ and the weight of $exit_i$ is 1.
	So the probability of a correct simulation is 
	$\frac{n^2+1}{n^2+2}= 1-\frac{1}{n^2+2}\geq e^{-\frac{2}{n^2+2}}$.~\footnote{We use $1-x\geq e^{-2x}$ for $0\leq x\leq \frac 1 2$.}
	\item If $i$ is a test of $c_j$ and the marking of $c_j$ is non null 
	then the weight of $dec_i$ is $2n^4+2$, the weight of $begZ_i$ is $n^2+1$ and the weight of $exit_i$ is 1.
	So the probability of a correct simulation is 
	$\frac{2n^4+2}{2n^4+n^2+4}\geq \frac{2n^4+2}{2n^4+2n^2+4}=  \frac{n^2+1}{n^{2}+2}= 1-\frac{1}{n^2+2}\geq e^{-\frac{2}{n^2+2}}$.
	\item If $i$ is a test of $c_j$ and the marking of $c_j$ is null 
	then the weight of $begZ_i$ is $n^2+1$ and the weight of $exit_i$ is 1.
	So the probability of a correct simulation is 
	$\frac{n^2+1}{n^2+2}= 1-\frac{1}{n^2+2}\geq e^{-\frac{2}{n^2+2}}$.
\end{itemize}
So the probability of the correct simulation during the $(n+1)^{th}$ simulation is at least $(e^{-\frac{2}{n^2+2}})^\ell=e^{-\frac{2\ell}{n^2+2}}$.
Hence the probability of $\sigma$ is at least $2^{-\ell}\prod_{n\in \nat}e^{-\frac{2\ell}{n^2+2}}=2^{-\ell}e^{-\sum_{n\in \nat} \frac{2\ell}{n^2+2}}>0$, 
as the sum in the exponent converges.
Let this probability be $\theta_0$.
Observe that the set of markings visited during the simulation (say $B$) is infinite since after every simulation, $sim$ is incremented.
Moreover by definition of $\sigma$, for all $\mathbf{m}\in B$, 
$\prob_{\mathcal N_\mathcal C,\mathbf{m}_0}(\mathbf{F}\mathbf{m})  \geq \theta_0$ and $\prob_{\mathcal N_\mathcal C,\mathbf{m}}(\mathbf{F} A) \geq \theta_0$.

\smallskip\noindent
Assume by contradiction that  $\mathcal N_\mathcal C$ is divergent w.r.t. to $\mathbf{m}_0$ and $A$
and consider the corresponding functions $f_0$ and $f_1$. Then $B\subseteq  \{\mathbf{m} \mid f_0(\mathbf{m}) \geq \theta_0 \wedge  f_1(\mathbf{m}) \geq \theta_0\} \cap Reach(\mathcal N_\mathcal C,\mathbf{m}_0)$
falsifying the third condition of divergence.

\smallskip\noindent
$\bullet$ Assume that $\mathcal C$ does not halt.
 Let us bound the probability that a reachable marking
 $\mathbf{m}$ fufills $\mathbf{m}(sim)\geq n$. Recalling that to mark again $p_0$
 after a necessarily faulty simulation, when $sim$ is marked implies that $rm_i$ can fire, the probabilistic behaviour
 of the marking $\mathbf{m}(sim)$ after every firing of $again$ or some $exit_i$ until it possibly reaches $n$ 
 or remains unchanged is depicted as follows.
 \begin{center}
  \begin{tikzpicture}[node distance=2cm,->,auto,-latex]

  \path (1,0) node[minimum size=0.6cm,draw,circle,inner sep=2pt] (q1) {$\bot$};
   \path (2,0) node[minimum size=0.6cm,draw,circle,inner sep=2pt] (q2) {$1$};
   \path (3,0) node[] (q3) {$\cdots$};
    \path (4,0) node[minimum size=0.6cm,draw,circle,inner sep=2pt] (qi) {$i$};
  \path (5,0) node[] (q4) {$\cdots$};
   \path (6,0) node[minimum size=0.6cm,draw,circle,inner sep=2pt] (qn) {$n$};

 \draw[arrows=-latex'] (q2)-- (1.5,-0.5) node[pos=1,below] {$\frac 1 {7}$}--(q1) ;
  \draw[arrows=-latex'] (q2)-- (2.5,0.5) node[pos=1,above] {$\frac {6} {7}$}--(q3) ;

 \draw[arrows=-latex'] (qi)-- (3.5,-0.5) node[pos=1,below] {$\frac 2 {3} z_i$}--(q3) ;
   \draw[arrows=-latex'] (qi)-- (4.5,0.5) node[pos=1,above] {$\frac {1} {3} z_i$}--(q4) ;
   \draw[arrows=-latex'] (qi)-- (4,-2)-- (1,-2) node[pos=0.5,above] {$\frac {1} {2n^4+n^2+4} z_i$}--(q1) ;

  \path (6.5,-1) node[] (q4) {{\small where $z_i=(1+\frac {1} {2n^4+n^2+4} )^{-1}$}};

  \end{tikzpicture}
\end{center}
Here we have upper bounded the stochastic behaviour of $\mathbf{m}(sim)$ i.e.,
(1) we have not considered the possibility of multiple firings of $rm_i$
and (2) we have only considered the possibility of a single firing of $exit_i$
during a simulation. The absorbing state $\bot$ of this random walk represents the fact that $stop$ is marked
(and so the marking of $sim$ remains unchanged in the future).
This random walk can be stochastically bounded by this one:

 
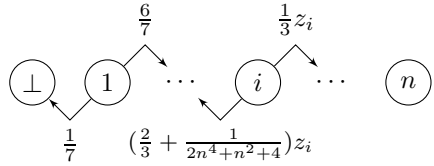
\begin{center}
  \begin{tikzpicture}[node distance=2cm,->,auto,-latex]

  \path (1,0) node[minimum size=0.6cm,draw,circle,inner sep=2pt] (q1) {$\bot$};
   \path (2,0) node[minimum size=0.6cm,draw,circle,inner sep=2pt] (q2) {$1$};
   \path (3,0) node[] (q3) {$\cdots$};
    \path (4,0) node[minimum size=0.6cm,draw,circle,inner sep=2pt] (qi) {$i$};
  \path (5,0) node[] (q4) {$\cdots$};
   \path (6,0) node[minimum size=0.6cm,draw,circle,inner sep=2pt] (qn) {$n$};

 \draw[arrows=-latex'] (q2)-- (1.5,-0.5) node[pos=1,below] {$\frac 1 {7}$}--(q1) ;
  \draw[arrows=-latex'] (q2)-- (2.5,0.5) node[pos=1,above] {$\frac {6} {7}$}--(q3) ;

 \draw[arrows=-latex'] (qi)-- (3.5,-0.5) node[pos=1,below] {{\small $(\frac 2 {3}+\frac {1} {2n^4+n^2+4}) z_i$}}--(q3) ;
   \draw[arrows=-latex'] (qi)-- (4.5,0.5) node[pos=1,above] {$\frac {1} {3} z_i$}--(q4) ;

  \end{tikzpicture}
\end{center}
and then by this one:
 
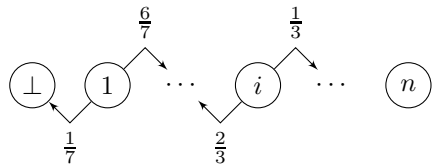
\begin{center}
  \begin{tikzpicture}[node distance=2cm,->,auto,-latex]

  \path (1,0) node[minimum size=0.6cm,draw,circle,inner sep=2pt] (q1) {$\bot$};
   \path (2,0) node[minimum size=0.6cm,draw,circle,inner sep=2pt] (q2) {$1$};
   \path (3,0) node[] (q3) {$\cdots$};
    \path (4,0) node[minimum size=0.6cm,draw,circle,inner sep=2pt] (qi) {$i$};
  \path (5,0) node[] (q4) {$\cdots$};
   \path (6,0) node[minimum size=0.6cm,draw,circle,inner sep=2pt] (qn) {$n$};

 \draw[arrows=-latex'] (q2)-- (1.5,-0.5) node[pos=1,below] {$\frac 1 {7}$}--(q1) ;
  \draw[arrows=-latex'] (q2)-- (2.5,0.5) node[pos=1,above] {$\frac {6} {7}$}--(q3) ;

 \draw[arrows=-latex'] (qi)-- (3.5,-0.5) node[pos=1,below] {$\frac 2 {3}$}--(q3) ;
   \draw[arrows=-latex'] (qi)-- (4.5,0.5) node[pos=1,above] {$\frac {1} {3}$}--(q4) ;

  \end{tikzpicture}
\end{center}
Now applying the standard results of gambler's ruin, one gets that the probability
that the marking of $sim$ reaches $n$ (for $n>2$) is bounded by: $\frac{6}{2^{n-2}-1}$.

\smallskip\noindent
We now define function $f_0$ and  show that it fulfills the hypotheses of Proposition~\ref{proposition:f01}. 
$f_0(\mathbf{m})$ is defined by: 
\begin{itemize}
	\item if $\mathbf{m}(sim)\leq 2$ then $f_0(\mathbf{m})=1$ 
	\item else $f_0(\mathbf{m})=2^{-(\mathbf{m}(c_1)+\mathbf{m}(c'_1))}+\frac{6}{2^{\mathbf{m}(sim)-2}-1}$. 
\end{itemize}

\smallskip\noindent
Let us verify the conditions of divergence. Let $0<\theta<1$.
\begin{small}
\begin{align*}
& \prob_{\mathcal N_\mathcal C,\mathbf{m}_0}\left(\mathbf{F }(f_0(\mathbf{m})\leq \theta)\right)\\
=\ & \prob_{\mathcal N_\mathcal C,\mathbf{m}_0}(\mathbf{F }(\mathbf{m}(sim)> 2 \\
& \hspace*{1.7cm}\wedge 2^{-(\mathbf{m}(c_1)+\mathbf{m}(c'_1))}+\frac{6}{2^{\mathbf{m}(sim)-2}-1}\leq \theta))\\
\leq\ & \prob_{\mathcal N_\mathcal C,\mathbf{m}_0}(\mathbf{F }((\mathbf{m}(sim)> 2 \wedge \frac{6}{2^{\mathbf{m}(sim)-2}-1}\leq \frac \theta 2)\\
&\hspace*{1.7cm}\vee (2^{-(\mathbf{m}(c_1)+\mathbf{m}(c'_1))}\leq \frac \theta 2)))\\
\leq\ & \prob_{\mathcal N_\mathcal C,\mathbf{m}_0}(\mathbf{F }(\mathbf{m}(sim)> 2 \wedge \frac{6}{2^{\mathbf{m}(sim)-2}-1}\leq \frac \theta 2))\\
&+\ \prob_{\mathcal N_\mathcal C,\mathbf{m}_0}(\mathbf{F }(2^{-(\mathbf{m}(c_1)+\mathbf{m}(c'_1))}\leq \frac \theta 2))\\
\leq\ & \frac \theta 2 + \frac \theta 2=\theta
\end{align*}
\end{small}
On the other hand $f_0(\mathbf{m})\geq \theta$ implies that  $\mathbf{m}(c_1)$, $\mathbf{m}(c'_1)$ and $\mathbf{m}(sim)$ are bounded by some constant.
Since for all reachable $\mathbf{m}$, $\mathbf{m}(c_2)+\mathbf{m}(c'_2)=\mathbf{m}(c_1)+\mathbf{m}(c'_1)$ and
$\sum \mathbf{m}(p_i)  + \sum \mathbf{m}(q_i) +\mathbf{m}(stop)=1$, all places are bounded and thus
the set of reachable markings  $\mathbf{m}$ such that  $f_0(\mathbf{m})\geq \theta$ is finite. 
\end{proof} 

\bigskip

\section{Divergent models by construction}
\label{sec:case}

Due to the undecidability results, in this section, we  propose syntactical restrictions for standard models like channel systems and
pushdown automata that ensure divergence.  The core  idea is to introduce a numerical distance between a state and the target
(below represented by a computable function $f$).  Then the sufficient conditions for divergence are stated as follows:
(1) On average given any current state, the distance increases at the next step at least by fixed amount 
($\varepsilon$ in the following theorems) and (2) the difference between the current and next distances is bounded by
a fixed quantity ($K$ in the following theorems).

\subsection{Sufficient conditions for divergence}
We have focused our study on the existence of the function $f_1$ but a similar work remains to be done for the function $f_0$.
Observing that function
$f_1$ of Definition~\ref{definition:divergent} is somewhat related to transience of Markov chains, we first recall a sufficient condition of transience, based on which a new one is established.  Then we derive a  sufficient condition of divergence for infinite Markov chains
used for our two illustrations.

\begin{theorem}[\cite{RW16}] 
\label{theorem:sufficient-transient}
Let $\mathcal M$ be a Markov chain and
 $f$ be a function from $S$ to $\real$ with $B=\{s \mid f(s)\leq 0\}$  fulfilling
 $\emptyset \subsetneq B \subsetneq S$ and $\varepsilon,K \in \real_{>0}$
such that:
\begin{small}
\begin{equation}
\begin{aligned}
\mbox{for all }s\in S\setminus B\qquad \sum_{s'\in S} \suiv(s,s')f(s')\geq f(s)+\varepsilon \mbox{ and }\\ \sum_{|f(s')-f(s)|\leq K} \suiv(s,s')=1
\end{aligned}
\label{equation:div}
\end{equation}
\end{small}
Then:
$$\mbox{for all }s\in S\setminus B \qquad \prob_{\mathcal M,s}({\bf F} B)\leq c_1e^{-c_2f(s)}$$
where $c_1=\sum_{n\geq 1}e^{-\frac{\varepsilon^2n}{2(\varepsilon+K)^2}}$ and $c_2=\frac{\varepsilon}{(\varepsilon+K)^2}$.
\noindent \\
Which implies  transience of $\mathcal M$ when it is irreducible.
\end{theorem}

Theorem~\ref{theorem:sufficient-transient} suffers two restrictions. 
First it requires that when $s\notin B$, after every step $f(s)$ is increased by $\varepsilon$ on
average. Second it restricts the maximal change of $f(s)$ to
be bounded by a constant $K$. We propose two extensions
of this theorem that can then be used to get useful sufficient
conditions for divergence.

\begin{restatable}{theorem}{sufficienttransientbi}
\label{theorem:sufficient-transientbis}
Let $\mathcal M$ be a Markov chain and
 $f$ be a function from $S$ to $\real$ with $B=\{s \mid f(s)\leq 0\}$  fulfilling
 $\emptyset \subsetneq B \subsetneq S$, $\varepsilon,K \in \real_{>0}$ and $d\in \nat^*$
such that:
\begin{small}
\begin{equation}
\label{eq:div2}
\begin{aligned}
\mbox{for all }s\in S\setminus B\qquad \sum_{s'\in S} \suiv^{(d)}(s,s')f(s')\geq f(s)+\varepsilon \mbox{ and }\\ \sum_{|f(s')-f(s)|\leq K} \suiv(s,s')=1
\end{aligned}
\end{equation}
\end{small}
Then for all $s\in S $ such that $f(s)>dK$,
$$\prob_{\mathcal M,s}({\bf F} B)\leq c_1e^{-c_2(f(s)-dK)}$$
where $c_1=\sum_{n\geq 1}e^{-\frac{\varepsilon^2n}{2(\varepsilon+K)^2}}$ and $c_2=\frac{\varepsilon}{(\varepsilon+K)^2}$.\\
which implies  transience of $\mathcal M$ when it is irreducible.
\end{restatable}
\begin{proof}
Consider $\mathcal M^{(d)}=(S, \suiv^{(d)})$.\\ 
Since $ \sum_{|f(s')-f(s)|\leq dK}\suiv^{(d)}(s,s')=1$, $\mathcal M^{(d)}$
fulfills the conditions of Theorem~\ref{theorem:sufficient-transient} and so: 
$$\mbox{for all }s\in S\setminus B \qquad \prob_{\mathcal M^{(d)},s}({\bf F} B)\leq c_1e^{-c_2f(s)}$$
From this result one gets for all $s$ such that $f(s)>\alpha$:
$$\prob_{\mathcal M^{(d)},s}({\bf F} \{s' \mid f(s')\leq \alpha \})\leq c_1e^{-c_2(f(s)-\alpha)}$$

Consider $d'\leq d$ transitions in $\mathcal M$, 
leading from $s\in S\setminus B$ to $s'$. Then $f(s)-f(s')\leq dK$ as long as $B$ is not reached.
So for all $s$, $\prob_{\mathcal M,s}({\bf F} B)\leq \prob_{\mathcal M^{(d)},s}({\bf F} \{s' \mid f(s')\leq dK\})$.
\end{proof}

\bigskip

From this theorem, one can get a sufficient condition for divergence. 

\begin{restatable}{proposition}{sufficientdivergentbis} 
\label{proposition:sufficient-divergentbis}
Let $\mathcal M$ be a Markov chain and
 $f$ be a computable function from $S$ to $\real$ with $B=\{s \mid f(s)\leq 0\}$  fulfilling
 $\emptyset \subsetneq B \subsetneq S$, and for some $\varepsilon,K \in \real_{>0}$ and $d\in \nat^*$,
 Equation~(\ref{eq:div2}). Assume in addition that for all $n\in \nat$, $\{s \mid f(s)\leq n\}$
 is finite. Then $\mathcal M$ is divergent w.r.t. any $s_0$ and any finite $A$. 
 \end{restatable}
 \begin{proof}
First observe that one can compute $c'_1\geq c_1$ and $0<c'_2\leq c_2$
so that  for all $s\in S $ such that $f(s)>dK$,
$\prob_{\mathcal M,s}({\bf F} B)\leq c_1e^{-c_2(f(s)-dK)}$.
Fix some $s_0$ and some finite $A$. Let $a=\max(0,\max_{s\in A}(f(s))$.

\smallskip\noindent
Let $f_a(s)=f(s)-a$. Then $A\cup B \subseteq \{s \in S \mid f_a(s)\leq 0\}$.
Since $f_a$ is a translation of $f$, Equation~(\ref{eq:div2}) is satisfied for all $s \in S\setminus B$
and thus for all $s \in S\setminus \{s \mid f_a(s)\leq 0\}$.
So one can apply Theorem~\ref{theorem:sufficient-transientbis} for $f_a$.
Let us define $g(s)=c'_1e^{-c'_2(f(s)-dK)}$ 
or more precisely some computable function such that $g(s)\leq 2c'_1e^{-c'_2(f(s)-dK)}$. Then
since $A\subseteq  \{s \mid f_a(s)\leq 0\}$, $\prob_{\mathcal M,s}({\bf F} A)\leq g(s)$.
Furthermore since for all $n\in \nat$, $\{s \mid f(s)\leq n\}$ is finite, we deduce that for all $\theta>0$,
$\{s \mid g(s)\geq \theta\}$ is finite. So $g$ fulfills the condition of Proposition~\ref{proposition:f01} and $\mathcal M$
is divergent w.r.t. $s_0$ and $A$. 
\end{proof}

\bigskip

Now we study a more general and more involved extension whose proof requires standard results of martingale theory (the proof is available in Appendix). 
Note that we do not use it in the illustration of the following sections. 
The interest of this extension, i.e.,  the next theorem, lies in the fact that $f$
can now decrease by \emph{a non constant factor} i.e., some $O(f(s)^\alpha)$ for $0<\alpha<\frac 1 2$. 
 Observe that  $\sum_{n\geq 1} {\gamma}^{n^{1-2\alpha}}<\infty$ since $0<\gamma<1$ and $1-2\alpha>0$.

\begin{restatable}{theorem}{sufficienttransientbis} 
\label{theorem:sufficient-transient2}
Let $\mathcal M$ be a Markov chain and
 $f$ be a function from $S$ to $\real$ with $B=\{s \mid f(s)\leq 0\}$  fulfilling 
 $\emptyset \subsetneq B \subsetneq S$, $0<\varepsilon\leq 1$, $K\geq 2$, $K'\geq 0$ and $0\leq \alpha <\frac{1}{2}$
such that:
\begin{small}
\begin{equation}
\label{eq:div3}
\begin{aligned}
\mbox{for all }s\in S\setminus B\qquad f(s)+\varepsilon \leq \sum_{s'\in S} p(s,s')f(s')
 \mbox{ and }\\ \sum_{-K-K'f(s)^\alpha \leq f(s')-f(s)\leq K} \hspace*{-1cm}p(s,s')=1
\end{aligned}
\end{equation}
\end{small}
Then for all $s$ such that $f(s) \geq 2$:
$$\prob_{\mathcal M,s}({\bf F} B)\leq  \beta^{f(s)^{1-2\alpha}} \sum_{n\geq 1} {\gamma}^{n^{1-2\alpha}}$$
where  $0<\beta,\gamma<1$ are defined by:\\
\begin{footnotesize}
$\beta= e^{-\left(\frac{2(\varepsilon+K)K^\alpha}{\varepsilon}+ \frac{4K'(\varepsilon+K)K^\alpha}{\varepsilon^{1+\alpha}(1+\alpha)}
+ \frac{2K'^2K^{2\alpha}}{\varepsilon^{1+2\alpha}}(1+2\alpha)\right)^{-1}}$ 
and $\gamma=\beta^{\varepsilon^{1-2\alpha}}$\\
\end{footnotesize}
which implies  transience of $\mathcal M$ when it is irreducible. 
\end{restatable}

\bigskip

From this theorem, one gets another sufficient condition for divergence that 
we state here without proof since it follows the same lines as the one of Proposition~\ref{proposition:sufficient-divergentbis}.

\begin{proposition}  
\label{proposition:sufficient-divergent2}
Let $\mathcal M$ be a Markov chain and
 $f$ be a computable function from $S$ to $\real$ with $B=\{s \mid f(s)\leq 0\}$  fulfilling
  $\emptyset \subsetneq B \subsetneq S$, and for some $0<\varepsilon\leq 1$, $K\geq 2$, $K'\geq 0$ and $0\leq \alpha <\frac{1}{2}$, 
 Equation~(\ref{eq:div3}). Assume in addition that for all $n\in \nat$, $\{s \mid f(s)\leq n\}$
 is finite. Then $\mathcal M$ is divergent w.r.t. any $s_0$ and any finite $A$. 
 \end{proposition}

\subsection{Probabilistic channel systems}
\label{subsec:pCS}

Channel systems are a classical model for protocols where components communicate 
asynchronously via FIFO channels~\cite{Brand1983,Bolch98}. Let $\Sigma$ be an alphabet, we use  letters $a,b,c,x,y$ for items in $\Sigma$ and $w$ for a word in $\Sigma^*$.
Let $w\in \Sigma^*$, then $|w|$ denotes its length and $\lambda$ denotes the empty word.
One denotes $\Sigma_{\lambda}=\Sigma \cup \{\lambda\}$.
$\Sigma^{\leq k}$ is the set of words over $\Sigma$ 
with length at most $k$.

Now we introduce a probabilistic variant of channel systems
particularly appropriate for the modelling of open queuing networks.
Here a special input channel $c_{in}$ (that works as a counter) only 
receives the arrivals of anonymous clients all denoted by $\$$ (item 1 of the next definition).
Then the service of a client  corresponds to a message circulating between the other channels
with possibly change of message identity until the message disappears (items 2 and 3).

\begin{definition}
A  \emph{probabilistic open channel system (pOCS)} $\mathcal S=(Q, Ch,\Sigma,\Delta,W)$ is defined by:
\begin{itemize}
	\item a finite set $Q$ of states;
	\item a finite set $Ch$ of channels, including $c_{in}$;
	\item a finite alphabet $\Sigma$ including  $\$$;
	\item a transition relation $\Delta\ \subseteq Q \times Ch \times  \Sigma_{\lambda}  \times Ch  \times  \Sigma_{\lambda} \times Q$
	that fulfills:
	\begin{enumerate}
		\item For all $q\in Q$, $(q,c_{in},\lambda ,c_{in},\$,q) \in \Delta$;
		\item For all $(q,c,a,c',a',q') \in \Delta$,  $a=\lambda \Rightarrow a'=\$ \wedge c=c'=c_{in}$;
		\item For all $(q,c,a,c',a',q') \in \Delta$,  $c\neq c_{in}\Rightarrow c'\neq c_{in}$;
	\end{enumerate}
	\item $W$ is a function from $\Delta \times (\Sigma^*)^{Ch}$ to $\rat_{>0}$.
\end{itemize}
\end{definition}

\noindent
{\bf Interpretation.}
We consider three types of transitions between configurations: sending messages to the input channel (i.e., $(q,c_{in},\lambda ,c_{in},\$,q)$) representing client arrivals; 
transferring messages between different channels (i.e., $(q,c,a,c',a',q')$ with $\lambda\not\in \{a,a'\}$ and $c'\neq c_{in}$) describing client services; 
and terminating message processing
(i.e., $(q,c,a,c',\lambda,q')$ with $a\neq \lambda$) meaning client departures.  All messages entering $c_{in}$ are anonymous (i.e., denoted by $\$$).
The left part of Figure~\ref{fig:pCS} is a schematic view of 
such systems. The left channel is $c_{in}$. All dashed lines represent message arrivals (to $c_{in}$) or departures. 
The solid lines model message transferrings.

\tikzset{
queue/.pic={
  \draw[line width=1pt]
    (0,0) -- ++(2.75cm,0) -- ++(0,-1cm) -- ++(-2.75cm,0);
   \foreach \Val in {1,...,#1}
     \draw ([xshift=-\Val*10pt]2.75cm,0) -- ++(0,-1cm); 
  }}

 \begin{figure}[h!]
 \begin{center}
\begin{tikzpicture}[scale=0.6] 
\path 
  (0,1.5cm) pic {queue=3}
   (-5,0cm) pic {queue=3}
  (0,-1.5cm) pic {queue=6}; 

\foreach \Pos in {-2,1}
{
  \draw[line width=1pt][->] (-1, -0.5+\Pos*0.1) -- (0,\Pos);   
 \draw[line width=1pt][dashed][->] (3.9,\Pos-0.1) -- (5.1, \Pos-0.1);   
}
\draw[line width=1pt][dashed][->] (-5.3, -0.8) -- (-3.9,-0.8);
\draw[line width=1pt][->] (1, -0.9) -- (1,-1.5);
\draw[line width=1pt][->] (3.5,-0.9)--(3.5, -0.3); 

\draw[line width=2pt][dotted] (1.5, -0.9) -- (3,-0.9); 
\end{tikzpicture}
 
\end{center}

\caption{A schematic view of pOCS}

\label{fig:pCS}
\end{figure}

The next definitions formalize the semantic of pOCS.

\begin{definition}
Let $\mathcal S$ be a pOCS, $(q,\nu)\in Q \times (\Sigma^*)^{Ch}$ be a configuration and $t=(q,c,a,c',a',q')\in \Delta$.
Then $t$ is enabled in $(q,\nu)$ if $\nu(c)=aw$ for some $w$.
The firing of $t$ in  $(q,\nu)$ leads to $(q',\nu')$ defined by:
\begin{itemize}
	\item if $c=c'$ then $\nu'(c)=wa'$ and for all $c''\neq c$, $\nu'(c'')=\nu(c'')$;
	\item if $c\neq c'$ then $\nu'(c)=w$, $\nu'(c')=\nu(c')a'$ and for all $c''\notin \{c,c'\}$, $\nu'(c'')=\nu(c'')$.
\end{itemize}
\end{definition}
As usual one denotes the firing by $(q,\nu) \xrightarrow{t} (q',\nu')$. Observe that from any configuration
at least one transition (a client arrival) is enabled.

\begin{definition}
Let $\mathcal S$ be a pOCS. Then the Markov chain $\mathcal M_\mathcal S=(S_\mathcal S, \suiv_\mathcal S)$
is defined by:
\begin{itemize}
	\item $S_\mathcal S=Q \times (\Sigma^*)^{Ch}$ is the set of configurations; 
	\item For all $(q,\nu) \in S_\mathcal S$ 
	let {\small $W(q,\nu)=\sum_{(q,\nu)\xrightarrow{t}(q',\nu')} W(t,\nu)$}. Then: \\
	for all {\footnotesize $(q,\nu)\xrightarrow{t}(q',\nu')$}, {\footnotesize $\suiv_\mathcal S((q,\nu),(q', \nu'))\!=\!\frac{W(t,\nu)}{W(q,\nu)}$}.
\end{itemize}

\end{definition}

The restrictions on pOCS w.r.t.standard channel systems, shortly CS, do not change the status of the reachability problem.

\begin{restatable}{proposition}{pocsundecidability} 
\label{proposition:pOCS-undecidability}
The reachability problem of pOCS is undecidable. 
\end{restatable}
\begin{proof}
In this proof, to precise that a transition occurs in a (pO)CS $\mathcal C$,
we add the subscript $\mathcal C$ to $\rightarrow$ (and to $\rightarrow^*$, its transitive closure).

We reduce the reachability problem of CS to the one of pOCS.
Let $\mathcal C$ be a CS where w.l.o.g.
$c_{in}\notin Ch$ and $\$ \notin \Sigma$. Then the pOCS $\mathcal C'$ is defined as follows:
\begin{itemize}
	\item $Q'=Q$; $Ch'=Ch\cup \{c_{in}\}$; $\Sigma'=\Sigma \cup \{\$\}$;
	\item For all $(q,c,?,a,q')\in \Delta$, $(q,c,a,c,\lambda,q')\in \Delta'$;	
	\item For all $(q,c,!,a,q')\in \Delta$, $(q,c_{in},\$,c,a,q')\in \Delta'$;	
	\item For all $q\in Q$, $(q,c_{in},\lambda,c_{in},\$,q) \in \Delta'$.
\end{itemize}
Consider the relations $\sim$ and $\sim_0$ between a configuration $(q,\nu)$ of $\mathcal C$
and a configuration $(q',\nu')$ of $\mathcal C'$ defined by (1)
$(q,\nu)\sim(q',\nu')$ if $q=q'$ and for all $c\in Ch$, $\nu(c)=\nu'(c)$, and (2)
$(q,\nu)\sim_0(q',\nu')$ if $(q,\nu)\sim(q',\nu')$ and $\nu'(c_{in})=\lambda$.

\noindent
Due to the definition of $\mathcal C'$, for all $(q,\nu)\sim_0(q',\nu')$,
$(q,\nu)\xrightarrow{t}_{\mathcal C}(q_1,\nu_1)$ there exists  
$(q'_1,\nu'_1)$ such that $(q_1,\nu_1)\sim_0(q'_1,\nu'_1)$
and $(q',\nu')\rightarrow^*_{\mathcal C'}(q'_1,\nu'_1)$. 

\noindent
Similarly, for all $(q,\nu)\sim(q',\nu')$,
$(q',\nu')\xrightarrow{t}_{\mathcal C'}(q'_1,\nu'_1)$ there exists  
$(q_1,\nu_1)$ such that $(q_1,\nu_1)\sim(q'_1,\nu'_1)$
and $(q,\nu)\rightarrow^*_{\mathcal C}(q_1,\nu_1)$. 

\noindent
Now consider $(q_0,\nu_0)$ (resp. $(q_f,\nu_f)$)  a configuration of $\mathcal C$
and $(q'_0,\nu'_0)$ (resp. $(q'_f,\nu'_f)$)  the single configuration of $\mathcal C'$
such that $(q_0,\nu_0)\sim_0(q'_0,\nu'_0)$ (resp. $(q_f,\nu_f)\sim_0(q'_f,\nu'_f)$).
Then due to the simulation properties above,  $(q_f,\nu_f)$ is reachable
from $(q_0,\nu_0)$ in $\mathcal C$ iff $(q'_f,\nu'_f)$ is reachable
from $(q'_0,\nu'_0)$ in $\mathcal C'$.
\end{proof}

\bigskip

As discussed in the introduction, when the number of clients exceeds some threshold,
the performances of the system  drastically decrease and thus the ratio of arrivals
w.r.t. the achievement of a task increase. We formalize it by introducing
\emph{uncontrolled} pOCS where the weights of transitions are constant except the ones
of client arrivals which are specified by positive non constant polynomials.
Let $\nu \in (\Sigma^*)^{Ch}$. Then $|\nu|$ denotes $\sum_{c\in Ch}|\nu(c)|$. 

\begin{definition} Let $\mathcal S$ be a pOCS. Then $\mathcal S$ is \emph{uncontrolled} if:
\begin{itemize}
	\item For all $t=(q,c,a,c',a',q')\in \Delta$ with $a\neq \lambda$, $W(t,\nu)$ only depends on $t$
	and will be denoted $W(t)$;
	\item For all $t=(q,c_{in},\lambda,c_{in},\$,q)$,  $W(t,\nu)$
is a positive non constant polynomial, whose single variable is $|\nu|$,
and will be denoted $W_{in}(q, |\nu|)$.
\end{itemize}
\end{definition}
  
The next proposition establishes that an uncontrolled pOCS generates a divergent Markov chain.
This model illustrates the interest of divergence: while reachability of a pOCS is undecidable, 
we can apply Algorithm~\ref{algo:prob-reach-divbis}. 

\begin{restatable}{proposition}{divuncontrolledocs} 
\label{proposition:divergence-uncontrolled}
Let  $\mathcal S$ be a uncontrolled pOCS. Then $\mathcal M_\mathcal S$ is divergent.  
\end{restatable}
\begin{proof} 
We show that the conditions of Proposition~\ref{proposition:sufficient-divergentbis} are satisfied with $d=1$.
Let $n_0$ be an integer such that for all $n\geq n_0$,\\
$$\min_{q \in Q}(W_{in}(q, n))\geq 1+ \sum_{t =(q,c,a,c',a',q')\in \Delta\wedge a'\neq \$} W(t)$$
Let $f(q,\nu)=|\nu|$. Then Equation~\ref{eq:div2} is satisfied with $\varepsilon=\frac{1}{1+2 \sum_{t =(q,c,a,c',a',q')\in \Delta\wedge a'\neq \$} W(t)}$, 
$d=1$ and $K=1$.
Furthermore let $n\in \nat$, $\{(q,\nu)\mid |\nu|\leq n\}$ is finite. Thus $\mathcal M_\mathcal S$ is divergent.
\end{proof}

\bigskip

\noindent
{\bf Illustration.}
Open queueing network models combined with numerical solution methods are a very powerful paradigm for the 
performance evaluation of production systems.
Let us consider a simple production system~\cite{Bolch05} composed by the following components:
\begin{itemize}
	\item  a load station where the workpieces are mounted on to the pallet (LO);
	\item  some identical lathes (LA);
	\item  some identical milling machines (M);
	\item  a station to unload the pallet that removes the workpieces from the system (U).
\end{itemize} 
This open queueing network model is depicted in Figure~\ref{fig:App}. 

 \begin{figure}[h!]
 \begin{center}
\begin{tikzpicture}[scale=0.6] 
\path 
  (1.3,2cm) pic {queue=3} 
   (-5,0cm) pic {queue=6}
  (1.3,-2cm) pic {queue=3}
    (8.7,0cm) pic {queue=5};
  
\node at (-2.5,-0.8cm) (nodeLO) {LO};
\node at (3,1.2cm) (nodeM) {M};
\node at (3,-2.7cm) (nodeLA) {LA};
\node at (9.7,-0.8cm) (nodeU) {U};

\draw[line width=1pt][dashed][->] (12.9,-0.8) -- (15, -0.8); 
\draw[line width=1pt][dashed][->] (-6.3, -0.8) -- (-4.5,-0.8);

\draw[line width=1pt][->]  (-0.5, 0) -- (1,1.5)  -- (2,1.5) ;
\draw[line width=1pt][->]  (-0.5, -1.7) -- (1,-3.2)  -- (2,-3.2) ;
\draw[line width=1pt][->]  (6, -2.5) -- (7,-2.5) -- (7,-1.5) -- (1,0) -- (1,1)-- (2,1)  ;
\draw[line width=1pt][->]  (6, 0.7) -- (7,0.7) -- (7,-0.5) -- (1,-1.5) -- (1,-2.5)-- (2,-2.5)  ;
\draw[line width=1pt][->]  (6, 1.5)  -- (7.5,1.5)-- (7.5,-0.5)  -- (9,-0.5) ;
\draw[line width=1pt][->]  (6, -3.3)  -- (7.5,-3.3)-- (7.5,-1.2)  -- (9,-1.2) ;

\end{tikzpicture}
 
\end{center}

\caption{A simple production system}

\label{fig:App}
\end{figure}

In this system, the workpieces are mounted on the pallet in the load station $LO$, which 
are then transferred either to the milling machines $M$ or to the lathes $LA$. The workpieces leaving 
$M$ are passed either to the station that unloads the pallet $U$ or to the lathes $LA$. 
Similarly, the workpieces leaving $LA$ are passed either to $U$ or to $M$.

This system becomes uncontrolled pOCS if the weights of transitions representing workpieces 
arrivals are specified by positive non constant polynomials whose single variable is the total number of 
workpieces in the system, and all other transitions have constant weight.

\subsection{Probabilistic pushdown automata}

Pushdown automata are finite automata equipped with a stack, more expressive than finite automata but less  
than Turing machines~\cite{SCHUTZENBERGER1963}. 
We now introduce a subclass of probabilistic  pushdown automata that are divergent by construction.

\begin{definition}[pPDA] 
A \emph{(dynamic-)probabilistic pushdown automaton} (pPDA) is a tuple $\mathcal A= (Q, \Sigma,\Delta,W)$
where:
\begin{itemize}
	\item $Q$ is a finite set of control states;
	\item $\Sigma$ is a finite stack alphabet with $Q \cap \Sigma= \emptyset$;
	\item $\Delta$ is a subset of $Q \times \Sigma^{\leq 1} \times Q \times \Sigma^{\leq 2}$
	such that for all $(q,\lambda,q',w) \in \Delta$, $|w| \leq 1$;
	\item $W: \Delta \times \Sigma^* \to \rat_{> 0}$ is a computable function. 
\end{itemize}
\end{definition}

Observe that  within the version of PDA presented here, the size of the stack can vary by at most one when firing a transition.
In the version of pPDA presented in~\cite{Esparza06}, the weight function $W$ goes from $\Delta$ to $\rat_{>0}$.
In order to emphasize this restriction here and later we say that, in this case, the weight function is \emph{static} and the corresponding models
will be called static pPDA.  In what follows, pPDA denotes the dynamic version.
$W$ may also be seen as a finite set of functions $W_t: \Sigma^* \to \rat_{> 0}$, for all $t \in \Delta$.

As specified in the next definition, a transition $(q,\lambda,q',w) \in \Delta$ is only enabled when the stack is empty. An item $(q,a,q',w)$ of $\Delta$ is also denoted $q\xrightarrow{?a!w} q'$, 
the label $?\lambda !a$ is also simply denoted $!a$ while for $a\neq \lambda$ the label $?a!\lambda$ 
 is denoted by $?a$.
A \emph{configuration} of $\mathcal A$ is a pair $(q,w)\in Q \times \Sigma^*$.

\begin{definition}
Let  $\mathcal A= (Q, \Sigma,\Delta,W)$ be a pPDA. Then the Markov chain $\mathcal M_\mathcal A=(S_\mathcal A, \suiv_\mathcal A)$ is a pair where $S_\mathcal A=Q \times \Sigma^*$ and $\suiv_\mathcal A$ is defined by the four following rules,
 for all $q \in Q$ and $a\in \Sigma$:
\begin{itemize}
	\item If $\{t=q\xrightarrow{!w_t} q'\}_{t\in \Delta}=\emptyset$, one has $\suiv_\mathcal A( (q,\lambda),(q,\lambda))=1$;
	\item If $\{t=q\xrightarrow{!w_t} q'\}_{t\in \Delta}\neq\emptyset$,
	let {\small $V(q,\lambda)=\sum_{\hspace*{-0.5cm}t=q\xrightarrow{!w_t} q'} W(t,\lambda)$}. \\ Then:
        for all {\footnotesize $t\!=\!q\xrightarrow{!w_t} q'\!\!\in\!\Delta$}, one has {\footnotesize $\suiv_\mathcal A((q,\lambda),(q',w_t))\!=\!\frac{W(t,\lambda)}{V(q,\lambda)}$}
	\item If $\{t=q\xrightarrow{?a!w_t} q'\}_{t\in \Delta}=\emptyset$, one has
	$\suiv_\mathcal A( (q,wa),(q,wa))=1$;
	\item If $\{t=q\xrightarrow{?a!w_t} q'\}_{t\in \Delta}\neq\emptyset$,
	let {\small $V(q,wa)=\sum_{\hspace*{-0.5cm}t=q\xrightarrow{?a!w_t} q'} W(t,wa)$}. \\ Then:
	for all {\footnotesize $t\!=\!q\xrightarrow{?a!w_t} q'\!\!\in\!\Delta$}, one has {\footnotesize $\suiv_\mathcal A((q,wa),(q',ww_t))\!=\!\frac{W(t,wa)}{V(q,wa)}$}
\end{itemize}

\end{definition}

Let us define the \emph{increasing reachability relation} between states in $Q\times \Sigma$. For $(q,a), (q',a')\in Q\times \Sigma$, we say that $(q',a')$ is \emph{h-reachable} 
(i.e. reachable without changing the size of the stack) from $(q,a)$
if either $(q,a)=(q',a')$ or there is a sequence of transitions $(t_i)_{0\leq i<d}$ of $\Delta$ such that: 
$t_i\!=\!q_i\xrightarrow{?a_i!a_{i+1}} q_{i+1}$, 
$(q_0,a_0)\!=\!(q,a)$, 
$(q_d,a_{d})\!=\!(q',a')$ and for all $i$ one have $a_i\neq \lambda$.

We introduce the subset of \emph{increasing pairs}, denoted as  $Inc(\mathcal A) \subseteq Q\times \Sigma$ that contains pairs $(q,a)$ such that from every configuration
$(q,wa)$ with $w \in \Sigma^*$, the height of the stack can increase without decreasing before. 
When some conditions on  $Inc(\mathcal A)$ are satisfied, we obtain a syntactic \emph{sufficient} condition for $\mathcal M_\mathcal A$
to be divergent.

\begin{definition}
 
\noindent
The set of \emph{increasing pairs} $Inc(\mathcal A)\subseteq Q\times \Sigma$ of a pPDA  $\mathcal A= (Q, \Sigma,\Delta,W)$ is the set of pairs $(q,a)$
that can h-reach a pair $(q',a')$ 
with some $q'\xrightarrow{?a'!bc} q''\in \Delta$.
\end{definition}

Let us remark that $Inc(\mathcal A)$ can be easily computed in polynomial time by a saturation algorithm.

\begin{definition}
A pPDA $\mathcal A$ is \emph{increasing} if:
\begin{itemize}
	\item $Inc(\mathcal A)= Q\times \Sigma$;
	\item for all $t=q\xrightarrow{?a!w} q' \in \Delta$ such that $|w|\leq 1$,
	$W_t$ is an integer and we write $W_t \in \nat$;
	\item for every transition $t=q\xrightarrow{?a!bc} q' \in \Delta$,
	$W_t[X]$ is a non constant positive integer polynomial where its single variable $X$ is the height of the stack;
	\item for every transition $q\xrightarrow{?a} q' \in \Delta$, there exists a transition $q\xrightarrow{?a!bc} q'' \in \Delta$.
\end{itemize}

\end{definition}

\tikzset{
queue/.pic={
  \draw[line width=1pt]
    (0,0) -- ++(2.75cm,0) -- ++(0,-1cm) -- ++(-2.75cm,0);
   \foreach \Val in {1,...,#1}
     \draw ([xshift=-\Val*10pt]2.75cm,0) -- ++(0,-1cm); 
  }}

 \begin{figure}[h!]
 \begin{center}
\begin{tikzpicture}[xscale=0.36,yscale=0.39]

\draw[black](10,2) -- (10,-2) -- (-3,-2) -- (-3,2) --cycle;

  \path (-1,0) node[minimum size=0.6cm,draw,circle,inner sep=2pt] (q0) {$q_0$};
   \path (8,0) node[minimum size=0.6cm,draw,circle,inner sep=2pt] (q1) {$q_f$};

   \draw[-latex'] (3.6,2) .. controls +(55:120pt) and +(125:120pt) .. (3.4,2) node[pos=.2,right] {$?x!xy$};

\end{tikzpicture}
\end{center}
\caption{A schematic view of a pPDA modelling a server }

\label{fig:pCS}
\end{figure}

\smallskip\noindent
{\bf Illustration.} Figure~\ref{fig:pCS} is an abstract view of a pPDA modelling a server simultaneously handling multiple requests.
The requests may occur at any time and are stored in the stack. The loop labelled by $?x!xy$ is an abstract representation
of several loops: one per triple $(q,x,y)$ with $q\in Q$, $x\in \Sigma$ and  $y \in \Sigma$. 
Due to the abstract loop, the set of increasing pairs of the $pPDA_{server}$ is equal to $Q\times \Sigma$
and there is always a transition increasing the height of the stack outgoing from any $(q,a)$.
Assume now that 
any transition $t=q\xrightarrow{?a!ab} q$ has weight $W_t[X]=c_t X$ where $c_t\in \nat_{>0}$
and for any other transition, its weight does not depend on the size of the stack.
Then $\mathcal A$ is increasing.
The dependance on $X$ means that due to congestion, the time to execute tasks of the server
increases with the number of requests in the system and thus increase the probability of a new
request that occurs at a constant rate. 
One is interested to compute
the probability to reach $(q_f,\lambda)$ from $(q_0,\lambda)$ representing  the probability that the server 
reaches an idle state having served
all the incoming requests. 

The next proposition establishes that an increasing pPDA generates a divergent Markov chain.

\begin{restatable}{proposition}{incdivpda} 
\label{pro:divPDA}
Let  $\mathcal A$ be an increasing pPDA. Then the Markov chain $\mathcal M_\mathcal A$ is divergent
w.r.t. any $s_0$ and finite $A$. 
\end{restatable}
\begin{proof}
Let us denote $B_\mathcal A=\max(W_t\mid t=q\xrightarrow{?a!w} q' \in \Delta\wedge |w|\leq 1)$.
Due to the condition on $W$
for increasing $\mathcal A$, there exists some $n_0$ such that for all $n\geq n_0$,\\ 
\centerline{$2|\Delta|B_\mathcal A\leq \min(W_t(n)\mid t=q\xrightarrow{?a!bc} q'\in \Delta)$}

\smallskip\noindent
Let $f(q,w)=|w|$ and $d=|Q||\Sigma|$.
Let us show by induction on $d'\leq d$ that for all $(q,w)$ such that $|w|\geq n_0+d'$,
\begin{equation}
\label{eq:ppda1}
\sum_{(q',w')}\suiv^{(d')}((q,w),(q',w'))f(q',w')\geq f(q,w)
\end{equation}
and if a pair $(q',a')$ with a transition $q'\xrightarrow{?a'!bc} q''$ is h-reachable from pair $(q,w[|w|])$ in less than $d'$ steps then: 
\begin{equation}
\label{eq:ppda2}
\sum_{(q',w')} \suiv^{(d')}((q,w),(q',w'))f(q',w')\geq f(q,w)+\frac{1}{3B_\mathcal A^{d'}}
\end{equation}
\noindent
The basis case $d'=0$ is immediate. Assume the property is satisfied for some $d'$. There are two cases to be considered.
Let $a=w[|w|]$. Either (1) there is no $q \xrightarrow{?a} q'\in \Delta$ or (2)  there is some $q \xrightarrow{?a} q'\in \Delta$.

\noindent
{\bf Case (1)} Since the size of the stack cannot decrease after a step and all successors of $(q,w)$ satisfy 
Equation~(\ref{eq:ppda1}) w.r.t. $d'$, Equation~(\ref{eq:ppda1}) is satisfied for $(q,w)$ w.r.t. $d'+1$.
Assume there exists  a pair $(q',a')$ with a transition $q'\xrightarrow{?a'!bc} q''$ h-reachable from $(q,a)$ in less than $d'+1$ steps.
Then there is a successor $(q_1,w_1)$ of $(q,w)$ such that $(q',a')$ is h-reachable from $(q_1,w_1[|w_1|])$ in less than $d'$ steps.
Thus Equation~(\ref{eq:ppda2}) is satisfied for $(q_1,w_1)$ w.r.t. $d'$.\\ 
Since $\suiv((q,w),(q_1,w_1))\geq \frac{1}{B_\mathcal A}$,
Equation~(\ref{eq:ppda2}) is satisfied for $(q,w)$ w.r.t. $d'+1$.

\noindent
{\bf Case (2)}  There is some $q \xrightarrow{?a} q''\in \Delta$. 
So  there is some $q \xrightarrow{?a!bc} q'\in \Delta$ implying that there is a successor $(q',w')$ of $(q,w)$
with $|w'|=|w|+1$. Every other successor $(q_1,w_1)$ fulfill $|w_1|\geq |w|-1$. Whatever the successor, they all
satisfy Equation~(\ref{eq:ppda1}) w.r.t. $d'$. Since $n\geq n_0$, $\suiv((q,w),(q',w'))\geq \frac{2}{3}$,
$$\sum_{(q'',w'')} \suiv^{(d'+1)}((q,w),(q',w'))f(q'',w'')\geq f(q,w)+\frac{1}{3}$$
which achieves the proof of this case.

\noindent
By definition of increasing pPDA, from every $(q,a)$ there exists a pair 
$(q',a')$ with a transition $q'\xrightarrow{?a'!bc} q''$ h-reachable from pair $(q,a)$ in less than $d$ steps.
So for all $(q,w)$ such that $|w|\geq n_0+d$,
$$\sum_{(q',w')} \suiv^{(d)}((q,w),(q',w'))f(q',w')\geq f(q,w)+\varepsilon$$
with $\varepsilon=\frac{1}{3B_\mathcal A^{d}}$.

\noindent
Observe that in $d$ steps, the size of the stack can change by at most $d$ and that the number of configurations with size of stacks bounded 
by some $n$ is finite. So the hypotheses of Proposition~\ref{proposition:sufficient-divergentbis} are satisfied and $\mathcal M_\mathcal A$ is divergent. 
\end{proof} 

\bigskip

\section{Conclusion and perspectives}
\label{sec:conclusion}

We have introduced the divergence property of infinite Markov chains and designed two generic CRP-algorithms depending
on the status of the reachability problem. 
Then we have studied the decidability of divergence for pPDA and for pPNs
and different kinds of weights and target sets.  
Finally, we have provided two useful classes of divergent models within pOCS and pPDA.

\smallskip
One ongoing work is to integrate the generic algorithms into the tool Cosmos~\cite{BallariniBDHP15}, instantiated with our illustrating divergent models.
In the future, we plan to study 
the model checking of polynomial pPDA (as a possible extension of~\cite{Esparza06}) and to propose some heuristics to find functions $f_0$ and $f_1$. 
Furthermore, it is worthwhile to exhibit the classes of 
models for which decisiveness and divergence are complementary, such as random walks. 
\bibliographystyle{fundam}
\bibliography{RP}

\section{Appendix}
\label{sec:apprendix}

In this appendix, we prove Theorem~\ref{theorem:sufficient-transient2}. 
To do this, we first recall some definitions related to martingale theory.
Let $(\Omega,\mathcal F)$ and 
$(S,\mathcal E)$ be sets equipped with a $\sigma$-algebra.
A (discrete-time) stochastic process is a sequence $(X_n)_{n\in \nat}$ of random variables
from $(\Omega,\mathcal F)$ to  $(S,\mathcal E)$. A filtration of $\mathcal F$ is
a non decreasing sequence  $(\mathcal F_n)_{n\in \nat}$ of  $\sigma$-algebras included in $\mathcal F$.
A stochastic process $(X_n)_{n\in \nat}$
is adapted to a filtration $(\mathcal F_n)_{n\in \nat}$ if for all $n\in \nat$, $X_n$ is $\mathcal F_n$-measurable.
Let $\tau$ be a random variable from  $(\Omega,\mathcal F)$ to $\nat \cup \{\infty\}$.
$\tau$ is a stopping time if for all $n\in \nat$, the event $\tau=n$ belongs to $\mathcal F_n$.
Let $n\in \nat$, we denote $n\wedge \tau$, the random variable $\min(n,\tau)$.

\begin{definition} Let  $(X_n)_{n\in \nat}$ be a real-valued stochastic process
adapted to a filtration $(\mathcal F_n)_{n\in \nat}$.   $(X_n)_{n\in \nat}$ 
is a submartingale if for all $n\in \nat$, $\esp(X_{n+1}-X_n\mid \mathcal F_n)\geq 0$.
\end{definition}

Below we recall a result about transience of Markov chains based on martingale theory. 
We will use this classic lemma in order to establish Theorem~\ref{theorem:sufficient-transient2}.
\begin{lemma}[Azuma-Hoeffding inequality\cite{RW16}]
\label{lemma:azuma}
Let $(X_k)_{k\in \nat}$ be a submartingale and $\{c_k\}_{k\geq 1}$ be positive reals such that
for all $k\geq 1$, $|X_k -X_{k-1}|\leq c_k$. Let $a> 0$. Then for all $n\geq 0$: 
$$\prob(X_n-X_0\leq -a\mid \mathcal F_0)\leq e^{\frac{-a^2}{2\sum_{k=1}^n c_k^2}}$$
\end{lemma}

\sufficienttransientbis*
\begin{proof}
Pick some $s_0\in  S\setminus B$. Due to Equation~(\ref{eq:div3}),
there exists $s_1$ with $f(s_1)\geq f(s_0)+\varepsilon$ and by induction for all $n\in \nat$,
there exists $s_n$  with $f(s_n)\geq f(s_0)+n\varepsilon$. Thus $\lim_{n\rightarrow \infty} f(s_n)=\infty$.
Assume that the inequality of the conclusion holds, there exists some $n_0$ with  
$\prob_{\mathcal M,{s_{n_0}}}({\bf F} B)<1$. If $\mathcal M$ is irreducible, this establishes transience
of $\mathcal M$.

\noindent
Let us establish this inequality.
Let $a=f(X_0)\geq 2$. Let $\tau$ be the stopping time  associated with the entry in $B$.
Let $Y_n=f(X_n)-n\varepsilon$. Equation~(\ref{eq:div3}) implies that 
$Y_{n\wedge \tau}$ is a submartingale. Furthermore Equation~(\ref{eq:div3})  implies
that:

\centerline{$Y_{n\wedge \tau}\leq a+nK$}

\noindent
and so 
$|Y_{n+1\wedge \tau}-Y_{n\wedge \tau}|\leq \varepsilon+K+K'(a+nK)^\alpha$.

\smallskip\noindent
Observe that:
\begin{align*}
\prob(\tau<\infty)&=&\prob(\bigcup_{n\geq 1} f(X_{n\wedge \tau})-f(X_0)\leq -a)\\
&=&\prob(\bigcup_{n\geq 1} Y_{n\wedge \tau}-Y_0\leq -a-n\varepsilon)\\
&\leq&  \sum_{n\geq 1} \prob(Y_{n\wedge \tau}-Y_0\leq -a-n\varepsilon)
\end{align*}
Applying Lemma~\ref{lemma:azuma}, one gets:
\begin{align*}
&\prob(\tau<\infty)\\
\leq&  \sum_{n\geq 1}e^{-\frac{(a+n\varepsilon)^2}{2\sum_{k=1}^n (\varepsilon+K+K'(a+kK)^\alpha)^2}}\\
= & \sum_{n\geq 1}e^{-\frac{(a+n\varepsilon)^2}{2n(\varepsilon+K)^2+ \sum_{k=1}^n 4(\varepsilon+K)K'(a+kK)^\alpha+ 2K'^{2}\sum_{k=1}^n (a+kK)^{2\alpha}}}\\
\end{align*}
Observe that  for $u \geq 0$ and $\theta,v>0$,\\ $\sum_{k=1}^n (u+vk)^\theta\leq \int_{1}^{n+1} (u+vx)^\theta dx \leq \frac{(u+v(n+1))^{1+\theta}}{v(1+\theta)}$.\\
Thus:
\begin{align*}
&\prob(\tau<\infty)\\
 \leq&\sum_{n\geq 1}e^{-\frac{(a+n\varepsilon)^2}{2n(\varepsilon+K)^2+ \frac{4K'(\varepsilon+K)}{K(1+\alpha)}(a+(n+1)K)^{1+\alpha}+ \frac{2K'^{2}}{K(1+2\alpha)}(a+(n+1)K)^{1+2\alpha}}}\\
\end{align*}
Since $n\leq \frac{a+n\varepsilon}{\varepsilon}$ and $a+(n+1)K\leq \frac{K}{\varepsilon}(a+n\varepsilon)$,
\begin{footnotesize}
\begin{align*}
&\prob(\tau<\infty)\\
\leq&\sum_{n\geq 1} 
e^{-\frac{(a+n\varepsilon)^2}{\frac{2(\varepsilon+K)^2(a+n\varepsilon)}{\varepsilon}+ \frac{4K'(\varepsilon+K)K^\alpha(a+n\varepsilon)^{1+\alpha}}{\varepsilon^{1+\alpha}(1+\alpha)}
+ \frac{2K'^2K^{2\alpha}(1+2\alpha)(a+n\varepsilon)^{1+2\alpha}}{\varepsilon^{1+2\alpha}}}}\\
\end{align*}
\end{footnotesize}
After dividing  numerator and  denominator by $(a+n\varepsilon)^2$,
\begin{align*}
&\prob(\tau<\infty)\\
\leq& \sum_{n\geq 1}e^{-\frac{1}{\frac{2(\varepsilon+K)}{\varepsilon(a+n\varepsilon)}+ \frac{4K'(\varepsilon+K)K^\alpha}{\varepsilon^{1+\alpha}(1+\alpha)(a+n\varepsilon)^{1-\alpha}}
+ \frac{2K'^2K^{2\alpha}(1+2\alpha)}{\varepsilon^{1+2\alpha}(a+n\varepsilon)^{1-2\alpha}}}}\\
\end{align*}
Since $a+n\varepsilon\geq 1$,
\begin{align*}
&\prob(\tau<\infty)\\
\leq& \sum_{n\geq 1}e^{-\frac{1}{\left(\frac{2(\varepsilon+K)}{\varepsilon}+ \frac{4K'(\varepsilon+K)K^\alpha}{\varepsilon^{1+\alpha}(1+\alpha)}
+ \frac{2K'^2 K^{2\alpha}}{\varepsilon^{1+2\alpha}}(1+2\alpha)\right)(a+n\varepsilon)^{2\alpha-1}}}\\
=& \sum_{n\geq 1}e^{-\frac{(a+n\varepsilon)^{1-2\alpha}}{\left(\frac{2(\varepsilon+K)K^\alpha}{\varepsilon}+ \frac{4K'(\varepsilon+K)K^\alpha}{\varepsilon^{1+\alpha}(1+\alpha)}
+ \frac{2K'^2K^{2\alpha}}{\varepsilon^{1+2\alpha}}(1+2\alpha)\right)}}\\
\end{align*}

\smallskip\noindent
So: 
\begin{align*}
\prob(\tau<\infty)&\leq& \sum_{n\geq 1} \beta^{(a+n\varepsilon)^{1-2\alpha}}\\
&\leq& \sum_{n\geq 1} \beta^{a^{1-2\alpha}+(n\varepsilon)^{1-2\alpha}}\\
&=& \left(\sum_{n\geq 1} \beta^{(n\varepsilon)^{1-2\alpha}}\right) \beta^{a^{1-2\alpha}}\\
&=& \left(\sum_{n\geq 1} \gamma^{n^{1-2\alpha}}\right) \beta^{a^{1-2\alpha}}
\end{align*}
\end{proof}

\end{document}